\documentclass[12pt,a4paper,twoside,final]{article}
\usepackage[utf8]{inputenc}
\usepackage[T1]{fontenc}
\usepackage[margin=2.5cm]{geometry}

\usepackage{bef_alex_nothm,todonotes}
\usepackage{lmodern}
\usepackage{amsmath,amsthm,amssymb}
\usepackage{fixmath}
\usepackage{enumitem}
\usepackage{upgreek}
\usepackage{amsfonts}
\usepackage{microtype}
\usepackage{float}
\usepackage{mathrsfs}

\usepackage[notcite,notref]{showkeys}

\usepackage{hyperref}
\hypersetup{
     colorlinks   = true,
     citecolor    = blue
}

\usepackage[nameinlink]{cleveref} 

\theoremstyle{plain}
\newtheorem{thm}{Theorem}[section]
\newtheorem{lem}[thm]{Lemma}
\newtheorem{pro}[thm]{Proposition}

\newtheorem{cor}[thm]{Corollary}
\theoremstyle{definition}

\newtheorem{rem}[thm]{Remark}

\numberwithin{equation}{section}

\newcommand{\DOM}{\operatorname{dom}}
\newcommand{\te}{\mathrm{t}}
\newcommand{\teu}{\underline{\mathbf{t}}}
\newcommand{\teo}{\overline{\mathbf{t}}}

\DeclareMathOperator*{\argmin}{arg\,min}
\newcommand{\PG}{\mathrm{PG}}
\newcommand{\AC}{\mathrm{AC}}

\newcommand{\EEE}{\color{black}}
\renewcommand{\ti}{{\times}}
\newcommand{\ulG}{\underline{G}}
\newcommand{\Mto}{\xrightarrow{\,\rmM\,}}
\newcommand{\Gto}{\xrightarrow{\Gamma}}
\newcommand{\Gweak}{\overset{\Gamma}{\weak}}
\DeclareMathOperator*{\Div}{div}
\newcommand{\inveps}{{\ts\frac{\ds1}{\ds\eps}}}

\begin{document}

\title{An existence result and evolutionary $\Gamma$-convergence for 
   perturbed gradient systems\thanks{The research was partially
     supported by DFG via SFB\,910 via the subprojects A5 and A8}}

\author{Aras Bacho\footnotemark[1], \ 
        Etienne Emmrich\footnotemark[1], \  
        Alexander Mielke\footnotemark[2]
}
\date{January 13, 2018}
\maketitle

\footnotetext[1]{Technische Universit\"at Berlin, Sekretariat MA 5-3,
  Straße des 17.\,Juni 136, 10623 Berlin, Germany.}
\footnotetext[2]{Weierstra\ss{}-Institut f\"ur Angewandte Analysis und
  Stochastik, Mohrenstra\ss{}e 39, 10117 Berlin and Institut f\"ur
  Mathematik, Humboldt-Universit\"at zu Berlin, Germany.}

\begin{abstract}
The initial-value problem for the perturbed gradient flow 
\begin{align*}
\label{eq:I.1}
\begin{cases}\quad
B(t,u(t)) \in \partial\Psi_{u(t)}(u'(t))+\partial \calE_t(u(t)) &
\text{ for  a.a. } t\in (0,T),\\
\quad u(0)=u_0, &
\end{cases}
\end{align*} 
with a perturbation $B$ in a \textsc{Banach} space $V$ is
investigated, where the dissipation potential $\Psi_u: V\rightarrow
[0,+\infty)$ and the energy functional $\calE_t:V\rightarrow
(-\infty,+\infty]$ are nonsmooth and supposed to be convex and
nonconvex, respectively. The perturbation $B:[0,T]\times V \rightarrow
V^*,\ (t,v)\mapsto B(t,v)$ is assumed to be continuous and satisfies a
growth condition. Under additional assumptions on the dissipation
potential and the energy functional, existence of strong solutions is
shown by proving convergence of a semi-implicit discretization scheme
with a variational approximation technique.
\end{abstract}

\section{Introduction}
\label{se:Intro}

The aim of this paper is to provide existence results for the
initial-value problem for the doubly nonlinear evolution inclusion
\begin{align}
  B(t,u(t)) \in \partial\Psi_{u(t)}(u'(t))+\partial \calE_t(u(t))
  &\quad \text{in } V^* \text{ for a.a. } t\in (0,T),
\end{align} 
with a continuous perturbation $B$ in the separable and reflexive real
\textsc{Banach} space $(V,\Vert\cdot\Vert)$, where $\partial \Psi_u$
and $\partial \calE_t$ denote the subdifferential of $\Psi_u$ and
$\calE_t$, respectively. The functional $\Psi_u$ is supposed to be a
dissipation potential for all $u\in \DOM(\calE_t)$, i.e., it is proper,
lower semicontinuous and convex with $\Psi_u(0)=0$ for all $u\in
\DOM(\calE_t)$. If the functionals $\Psi_u$ and $\calE_t$ are
\textsc{Fr\'echet} differentiable, the differential inclusion (1.1)
becomes the abstract evolution equation (also called doubly nonlinear
equation in \cite{ColVis90CDNE,Coll92DNEE})
\begin{align*}
 \rmD\Psi_{u(t)}(u'(t))=-\rmD\calE_t(u(t))+B(t,u(t)) &\quad \text{in } 
V^* \text{ a.e. in } (0,T),
\end{align*} 
where $\rmD\Psi_u$ and $\rmD\calE_t$ denote the \textsc{Fr\'echet} derivative
of $\Psi_u$ and $\calE_t$ respectively. The question arises why it is
interesting to study perturbed gradient systems. First of all, to
consider perturbed systems is sometimes important in order to describe
physical systems near or far from equilibrium properly. There are many
ways to incorporate the perturbation in the equation.

The most frequently used method is to consider an $\varepsilon$-family
of equations, where the occurring terms depend on the parameter
$\varepsilon$, and then to pass to the limit as
$\varepsilon\rightarrow \infty$, where the limit equation corresponds
to the unperturbed system. Another way to treat perturbed systems is
to use an additional term in the equations like the term $B_t$ in
\eqref{eq:I.1} or even a combination of both as in \cite{Miel16DEMM},
where the author considered the family of equations
\begin{align*}
 \rmD\Psi^\varepsilon_{u(t)}(u'(t))=-\rmD\calE^\varepsilon_t(u(t))+B^\varepsilon(t,u(t)) 
\end{align*} to derive results on the so-called evolutionary
$\Gamma$-convergence.

Second, \cite [p.\,235]{Miel16DEMM} highlights with an example that in
some cases it can be easier to treat a system with a nontrivial but
exact gradient structure $(X,\wt\calE, \wt \Psi)$ perturbed gradient
system $(V,\calE,\Psi,B)$ with a simpler energy $\calE$ and simpler
dissipation potentials $\Psi_u$.

While Section \ref{se:ExistResult} provides the main existence result
in Theorem \ref{th:MainExist}, we devote Section \ref{se:EGC} to the
question of evolutionary $\Gamma$-convergence of families 
$(V,\calE^\eps,\Psi^\eps,B^\eps)$ of perturbed gradient systems. This
provides a generalization of the results developed in
\cite{SanSer04GCGF, Serf11GCGF, Miel16EGCG} for exact gradient flows,
i.e.\ the case where $B^\eps \equiv 0$. Following the ideas in
\cite[Thm.\,4.8]{MiRoSa13NADN}, our Theorem \ref{th:EGC.main} shows that
under suitable technical assumptions, including convexity of
$\calE^\eps$, it is enough to establish $\calE_t^\eps \Gto \calE^0_t$
(strong $\Gamma$-convergence in $V$) and $\Psi^\eps_{u_\eps} \Mto
\Psi^0_{u_0}$ in $V$, where \textsc{Mosco} convergence means weak and strong
$\Gamma$-convergence.  

In Section \ref{se:Example} we show that the abstract result on
evolutionary $\Gamma$-convergence can be used for the homogenization
of quasilinear parabolic systems. For that application the \textsc{Mosco}
convergence $\Psi^\eps_{u_\eps} \Mto \Psi^0_{u_0}$ is too restrictive,
such that it is necessary to generalize it to situations where the
strong $\Gamma$-convergence $\Psi^\eps_{u_\eps} \Gto \Psi^0_{u_0}$ is
sufficient, see Corollary \ref{co:StrongGa}. Here we rely on an 
novel argument from \textsc{Liero-Reichelt} \cite{LieRei15?HCHT},
where the weak convergence of $u_\eps\weak u_0$ in $\rmW^{1,1}(0,T;V)$ is
circumvented by exploiting the strong convergence of the piecewise
affine interpolants $\wh u^\tau_\eps \to \wh u^\tau_0$ in
$\rmW^{1,1}(0,T;V)$ for $\eps \to 0$ and $\tau>0$ fixed. 
 
The general structure is that we provide a full and detailed proof of
the existence result in Section \ref{se:ExistResult}, where we use
\textsc{De Giorgi}'s minimization scheme using variational
interpolators. The result on the evolutionary $\Gamma$-convergence in
Section \ref{se:EGC} follows the same lines but is considerably
simpler as existence of solutions is assumed to be shown. Hence, for
getting an overview of the strategy in Section \ref{se:ExistResult} it
might be helpful to browse through the more compact proof of Theorem
\ref{th:EGC.main} first. This will facilitate the subsequent reading
of the full details in Section \ref{se:ExistResult}. In particular,
the elaborate time-discretization using \textsc{De Giorgi}'s
variational interpolants is only needed there.

\section{The main existence result}
\label{se:ExistResult}

Before making all the assumptions concerning the dissipation
potential, the energy functional and the perturbation, we need some
basic tools from convex analysis.

\subsection{Preliminaries and notation}
\label{su:PrelimNot}
In this section we collect some important notions and results on
convex analysis and $\Gamma$-convergence, which we need later on for
the existence result. First of all, we introduce the so-called
\textsc{Legendre-Fenchel} transform (or conjugate) $\Psi^*$ of a
proper, lower semicontinuous and convex functional $\Psi:V\rightarrow
(-\infty,+\infty]$ that is defined by
\begin{align*}
  \Psi^*(\xi):=\sup_{u\in V}\left \lbrace \langle \xi,u\rangle -
    \Psi(u)\right \rbrace, \quad \xi\in V^*,
\end{align*} where $\langle\cdot,\cdot\rangle$ denotes the duality
pairing between the \textsc{Banach} space $V$ and it's topological
dual space $V^*$. From the definition, the \textsc{Fenchel-Young}
inequality
\begin{align*}
 \langle \xi,u\rangle \leq  \Psi(u)+\Psi^*(\xi), \quad v\in V, \xi\in V^*,
\end{align*} 
immediately follows. It is also easy to check that the conjugate
itself is proper, lower semicontinuous and convex, see for example
\textsc{Ekeland} and \textsc{T\'emam} \cite{EkeTem74ACPV}. If, in
addition, $\Psi(0)=0$, then $\Psi^*(0)=0$ holds too. For a proper
functional $F:V\rightarrow (-\infty,+\infty]$, the
$($\textsc{Fr\'echet}$)$-subdifferential of $F$ is given by the
multivalued map $\partial F:V\rightarrow 2^{V^*}$ with
\begin{align*}
\partial F(u):=\left\lbrace \xi\in V^*: \ \liminf_{v\rightarrow u}
  \frac{F(v)-F(u)-\langle \xi,v-u\rangle}{\Vert v-u\Vert} \geq 0\right
\rbrace 
\end{align*} 
for all elements $u$ in the effective domain $\DOM(F):=\lbrace v\in V
\mid F(v)<+\infty \rbrace$ of $F$.  For convex and proper functions
$F$, it follows by simple calculations that the subdifferential of $F$
is given by
\begin{align*}
\partial F(u)=\left\lbrace \xi\in V^*: \ F(u)\leq F(v)+\langle
  \xi,u-v\rangle \quad \text{ for all  } v\in V \right\rbrace. 
\end{align*}
The following lemma gives a relation between the subdifferential of a
functional and it's \textsc{Legendre-Fenchel} transform.

\begin{lem}\label{le:Leg.Fen}
  Let $\Psi:V\rightarrow (-\infty,+\infty]$ be a proper, lower
  semicontinuous and convex functional and let $\Psi^*:V^*\rightarrow
  (-\infty,+\infty]$ be the \textsc{Legendre-Fenchel} transform of
  $\Psi$. Then for all $(u,\xi)\in V\times V^*$ the following
  assertions are equivalent:
\begin{itemize}
\item[$i)$]\quad $\xi\in \partial \Psi(u) \quad \text{in } V^*;$
\item[$ii)$]\quad  $u\in \partial \Psi^*(\xi)\quad \text{in } V;$
\item[$iii)$]\quad $\langle \xi, u\rangle=\Psi(u)+\Psi^*(\xi) \quad \text{in
  } \mathbb{ R}.$
\end{itemize}
\end{lem}
\begin{proof}
  \textsc{Ekeland} and \textsc{T\'emam} \cite[Prop.\,5.1 and
  Cor.\,5.2 on pp.\,21]{EkeTem74ACPV}.
\end{proof}

For the dissipation potentials $\Psi_u$ we need the notion of
$\Gamma$-convergence, see \cite{Dalm93IGC, Brai02GCB, Brai06HGC} (also called
epigraph convergence in \cite{Atto84VCFO}). We consider a functional
$\Psi:V\to (-\infty,\infty]$ and a sequence $(\Psi_n)_{n\in \N}$ of
functionals all of which are lower semicontinuous convex functionals.
The (strong) $\Gamma$-convergence $\Psi_n \Gto \Psi$ in $V$ is defined via 
\[
\Psi_n\Gto \Psi\ \Longleftrightarrow \ \left\{ \ba{cl} 
\text{(a) }&\ds\ v_n \to v \  \Longrightarrow \ \Psi(v) \leq
               \liminf_{n\to \infty} \Psi_n(v_n),\\[0.4em]
\text{(b) }&\ \forall\, \wh v\in V\ \exists\, (\wh v_n)_{n\in \N}: \ \
\wh v_n\to \wh v \text{ and }\Psi(v) \geq
               \limsup_{n\to \infty} \Psi_n(v_n).\\
  \ea \right. 
\]
Here (a) is called the (strong) liminf estimate, while (b) is called
the (strong) limsup estimate or the existence of recovery
sequences. Similarly we define the (sequential) weak
$\Gamma$-convergence $\Psi_n \Gweak \Psi$ in $V$ via 
\[
\Psi_n\Gweak \Psi\ \Longleftrightarrow \ \left\{ \ba{cl} 
\text{(a) }&\ds\ v_n \weak v \  \Longrightarrow \ \Psi(v) \leq
               \liminf_{n\to \infty} \Psi_n(v_n),\\[0.4em]
\text{(b) }&\ \forall\, \wh v\in V\ \exists\, (\wh v_n)_{n\in \N}: \ \
\wh v_n\weak \wh v \text{ and }\Psi(v) \geq
               \limsup_{n\to \infty} \Psi_n(v_n).\\
  \ea \right. 
\] 
If both convergences hold, then we say that $\Psi_n$ \textsc{Mosco}
converges to $\Psi$ and write $\Psi_n \Mto \Psi$. In \cite[pp.\,271]{Atto84VCFO}
the following fundamental relation between $\Gamma$-convergence and
the \textsc{Legendre-Fenchel} transform was established:
\begin{equation}
 \label{eq:Gcvg.sw*}
  \Psi_n \Gto \Psi \quad \Longleftrightarrow \quad \Psi_n^* \Gweak \Psi^*,
\end{equation}
which always holds on reflexive Banach spaces $V$ if all $\Psi_n$ and
$\Psi_n^*$ are nonnegative (as for our dissipation potentials).

\subsection{Semi-implicit variational approximation scheme}
\label{su:VarApprox}
The basic idea to show the existence of strong solutions to \eqref{eq:I.1} with an initial condition $u=u_0\in V$ is to construct a solution via a
particular discretization scheme, more precisely, with a semi-implicit
\textsc{Euler} method. The usual implicit Euler method does not work
since the equation \eqref{eq:I.1} does not possess the gradient flow structure due to the nonpotential perturbation. With our approach, it is
possible to construct time-discrete solutions via a variational
approximation scheme. To illustrate this let for $N\in
\mathbb{N}\backslash \lbrace 0\rbrace$
\begin{align}
\label{eq:II.1}
 I_\tau=\lbrace 0=t_0<t_1<\cdots< t_n=n\tau<\cdots t_N=T \rbrace
\end{align} be an equidistant partition of the time interval $[0,T]$ with step size $\tau:=T/N$, where we omit the dependence of $t_n$ on the step size $\tau$ for simplicity. The approximation of \eqref{eq:I.1} is then given by
\begin{align}
\label{eq:II.2}
B(t_{n-1},U_\tau^{n-1})\in \partial \Psi_{U_\tau^{n-1}}\left(\frac{U_\tau^{n}-U_\tau^{n-1}}{\tau}\right)+\partial \calE_{t_n}(U_\tau^{n}), \quad n=1,\cdots,N,
\end{align} where the values $U_\tau^{n}\approx u(t_n)$, which shall
approximate the exact solution of \eqref{eq:I.1} at $t_n$, are to determine. If both the dissipation potential and the energy functional are
\textsc{Fr\'echet}-differentiable the inclusion \eqref{eq:II.2} becomes the
equation
\begin{align}
\label{eq:II.3}
B(t_{n-1},U_\tau^{n-1})=
\rmD\Psi_{U_\tau^{n-1}}\left(\frac{U_\tau^{n}-U_\tau^{n-1}}{\tau}\right)+\rmD
\calE_{t_n}(U_\tau^{n}), \quad n=1,\cdots,N.
\end{align} 
It is now simple to see that the value $U_\tau^{n}$ can be
characterized as a solution of the \textsc{Euler-Lagrange} equation
associated to the map
\begin{align*}
v\mapsto \Upphi(\tau,t_{t_{n-1}},{U_\tau^{n-1}},B(t_{n-1},U_\tau^{n-1});v),
\end{align*} where 
\begin{align}
\label{eq:II.4}
  \Upphi(r,t,u,w;v)=
  r\Psi_{u}\left(\frac{v-u}{r}\right)+\calE_{t+r}(v)-\langle w,v\rangle
\end{align}
for $ r\in \mathbb{R}^{>0},t\in[0,T)$ with $r+t\in[0,T]$,  $u,v\in V$,
and $ w\in V^*$. In fact, we determine the value $U_\tau^{n}$ by
minimizing the functional $\Upphi$ in the variable $v\in V$ under
suitable conditions on the dissipation potential and the energy
functional. To assure that the value $U_\tau^{n}$ satisfies the
inclusion \eqref{eq:II.2} also in the nonsmooth case, which is in general not
true, we make an assumption to enforce property.

\subsection{Assumptions for the main existence result}
\label{su:AssumpExistRes}

We now collect the assumption on the perturbed gradient system
$\PG=(V,\calE,\Psi,B)$ for our existence result. They will be denoted
in via (2.En), (2.$\Uppsi$m), and (2.Bk). 

The assumptions for the energy functional are the following. 

\begin{enumerate}[label=(\thesection .E\alph*),
 leftmargin=3.2em]
\item \label{eq:cond.E.1}
 \textbf{Constant domain.} For all $t\in[0,T]$, the
  functional $\calE_t:V \rightarrow (-\infty,+\infty]$ is proper and lower
  semicontinuous with the time-independent effective domain $D\equiv
  \DOM(\calE_t)\subset V$ for all $t\in [0,T]$.
\item \label{eq:cond.E.2} \textbf{Compactness of sublevels.}  There exists
  $t^*\in [0,T]$ such that the functional $E_{t^*}$ has compact
  sublevels in $V$.
\item \label{eq:cond.E.3} \textbf{Energetic control of power.} For all $u\in D$,
  the power map $t\mapsto \calE_t(u)$ is continuous on $[0,T]$ and
  differentiable in $(0,T)$ and its derivative $\partial_t\calE_t$ is
  controlled by the function $\calE_t$, i.e., there exist $C>0$ such that
\begin{align*}
  \vert \partial_t \calE_t(u)\vert \leq C \calE_t(u)\quad \text{for all } t\in
  (0,T) \text{ and } u\in D.
\end{align*} 
\item \label{eq:cond.E.4}\textbf{Chain rule.} For every absolutely continuous
  curve $v\in \AC([0,T];V)$ and every \textsc{Bochner} integrable
  functions $\xi \in \rmL^1(0,T;V^*)$ such that
\begin{align*}
  \sup_{t\in[0,T]}\vert \calE_t(u(t))\vert <+\infty, \quad \xi(t)\in \partial \calE_t(u(t)) \quad \text{ a.e. in } (0,T),\\
  \int_0^T\Psi_{u(t)}(u'(t))\dd t <+\infty \quad \text{and} \quad
  \int_0^T\Psi^*_{u(t)}(\xi(t))\dd t <+\infty,
\end{align*} the map $t\mapsto \calE_t(u(t))$ is absolutely continuous on
$[0,T]$ and
\begin{align*}
  \frac{d}{\dd t}\calE_t(u(t))\geq \langle \xi(t),u'(t)\rangle + \partial_t
  \calE_t(u(t))\quad \text{a.e. in }(0,T).
\end{align*}
\item \label{eq:cond.E.5} \textbf{Strong-weak closedness.} For all $t\in[0,T]$ and all sequences $(u_n,\xi_n)_{n\in \mathbb{N}} 
     \subset V\ti V^*$ with $\xi_n\in \partial \calE^{\eps_n}_t(u_n)$ such that 
\begin{align*}
  u_n \rightarrow u\in V, \quad \xi_n \rightharpoonup \xi\in V^*,
  \quad \calE_t(u_n)\rightarrow \mathcal{E}\in \mathbb{R} \quad \text{and}
  \quad \partial_t \calE_t(u_n)\rightarrow \calP \in
  \mathbb{R}
\end{align*} as $n\rightarrow\infty$, we have the relations
\begin{align*} 
\xi \in \partial \calE_t(u), \quad \calP \leq \partial_t \calE_t(u)\quad
\text{and} \quad \mathcal{E}=\calE_t(u). 
\end{align*}
\end{enumerate}
We first give a few relevant comments on these assumptions that will
be important below. 

\begin{rem}\label{re:Assump.E} \mbox{}\vspace{-0.6em}

\begin{itemize}

\item[$i)$] From Assumption \ref{eq:cond.E.3} we deduce with
  \textsc{Gronwall}'s lemma the chain of inequalities
\begin{align}
\label{eq:II.5}
\ee^{-C\vert t-s\vert}\calE_s(u)\leq \calE_t(u)\leq \ee^{C\vert
  t-s\vert}\calE_s(u) \quad \text{ for all }s,t\in [0,T].
\end{align} 
In particular there exists a constant $C_1>0$ such that
\begin{align}
  \label{eq:II.6}
  G(u)=\sup_{t\in [0,T]}\calE_t(u)\leq C_1 \inf_{t\in[0,T]} \calE_t(u)
  \quad \text{for all } u\in D.
\end{align}

\item[$ii)$] From Assumptions \ref{eq:cond.E.2} and \ref{eq:cond.E.3}
  we deduce the existence of a real number $S$ which bounds the energy
  functional from below, i.e.,
\begin{align}
\label{eq:II.7}
\calE_t(u)\geq S \quad \text{for all }u\in V,\, t\in [0,T].
\end{align}

\item[$iii)$] From the strong-weak closedness property of the graph of
  $\partial E$ in \ref{eq:cond.E.5} and \textsc{Mordukhovich}
  \cite[Lem.\,2.32, p.\,214]{Mord06VAGD1} one can argue as in
  \cite[Prop.\,4.2, p.\,273]{MiRoSa13NADN}, in order to show the
  following variational sum rule:
  \\
  If for $u_0\in V,\ r>0$, and $t\in[0,T]$ the point $u\in V$ is a
  global minimizer of $\Upphi(\tau,t,u_0,w;\cdot)$, then 
\begin{align}
\label{eq:II.8}
\exists\, \xi \in \partial \calE_t(u):\quad w-\xi \in \partial
\Psi_{u_0}\left(\frac{u-u_0}{r}\right); 
\end{align} 
or equivalently $\ds w\in \partial
\Psi_{u_0}\left(\frac{u-u_0}{r}\right)+ \partial \calE_{t+r}(u)$. 

\item[$iv)$] Assumption \ref{eq:cond.E.2} and point $i)$ in this
  remark yields immediately that the functional $\calE_t$ has compact
  sublevels for all $t\in[0,T]$.
%
\item[$v)$] It is possible to relax Assumption \ref{eq:cond.E.3} by
  assuming not the time differentiability but a kind of
  \textsc{Lipschitz } continuity and a conditioned one-sided time
  differentiability of the map $t\mapsto \calE_t(u)$, see
  \cite{MiRoSa13NADN}. We shall confine ourselves to Assumption
  \ref{eq:cond.E.3} just to simplify the proofs.
 \end{itemize} 
\end{rem}

Now, we collect the assumptions concerning the dissipation potential $\Psi$.

\begin{enumerate}[label=(\thesection.$\Uppsi$\alph*), leftmargin=3.2em]
\item \label{eq:Psi.1} \textbf{Dissipation potential.} For all $u\in
  V$ the functional $\Psi_u: V\rightarrow [0,+\infty)$ is lower
  semicontinuous and convex with $\Psi(0)=0$. Furthermore if $w_1,w_2
  \in \partial \Psi_u(v)$ for any $v \in V$ then
  $\Psi_u^*(w_1)=\Psi_u^*(w_2)$.

\item \label{eq:Psi.2} \textbf{Superlinearity.} The functionals
  $\Psi_u$ and $\Psi_u^*$ are coercive uniformly with respect to $u\in
  V$ in sublevels of $E$, i.e., for all $R>0$ there hold
  \begin{align*}
    \lim_{\Vert \xi\Vert_*\rightarrow +\infty}\frac{1}{\Vert
      \xi\Vert_*}\Big(\inf_{\overset{u\in V}{G(u)\leq R}
    }\Psi^*_u(\xi)\Big)=\infty,\quad \lim_{\Vert v\Vert\rightarrow
      +\infty}\frac{1}{\Vert v\Vert}\Big(\inf_{\overset{u\in
        V}{G(u)\leq R} }\Psi_u(v)\Big)=\infty,
  \end{align*} 
  where $G(u):=\sup_{t\in[0,T]}\calE_t(u)$ for all $u\in V$.

\item \label{eq:Psi.3} \textbf{State-dependence is \textsc{Mosco}
    continuous.} The functional $\Psi$ is continuous in the sense of
  \textsc{Mosco}-convergence, i.e., for all $R>0$ and sequences
  $(u_n)_{n\in\mathbb{N}}\subset V$ with $u_n\rightarrow u\in V$ as
  $n\rightarrow\infty$ and $\sup_{n\in\mathbb{N}}G(u_n)\leq R$, we
  have $\Psi_{u_n} \Mto \Psi_u$.
\end{enumerate}

\begin{rem}\label{re:Assump.Psi} \mbox{} \vspace{-0.6em}
\begin{itemize}
\item[$i)$] Since $\DOM(\Psi_u)=V$ for all $u\in V$, the lower
  semicontinuity and convexity of $\Psi_u$ yields the continuity of
  $\Psi_u$ and $\partial \Psi_u(v)\neq \emptyset$ for all $u\in V,
  u\in D$. Together with Assumption \ref{eq:Psi.2}, this implies that the
  \textsc{Legendre-Fenchel} conjugate $\Psi^*$ is everywhere finite,
  i.e., $\DOM(\Psi^*)=V^*$, and the operator $\partial \Psi_u:
  V\rightarrow 2^{V^*}$ is for all $u\in D$ bounded, i.e., it maps
  bounded subsets of $V$ into bounded subsets of $V^*$. The former in
  turn entail the same properties for $\Psi^*_u$ for all $u\in V$.

\item[$ii)$] The \textsc{Mosco} convergence of $\Psi_{u_n} \Mto \Psi_u$
  from Assumption \ref{eq:Psi.3} implies \textsc{Mosco} convergence of the dual
  potentials, namely $\Psi^*_{u_n} \Mto \Psi^*_n$, see
  \eqref{eq:Gcvg.sw*}. In particular, this implies that for all
  $R>0$, all 
sequences $(u_n)_{n\in\mathbb{N}}\subset V$ with $u_n\rightarrow u\in
V$ and $\sup_{n\in \N} G(u_n)\leq R$, and all sequences $ (\xi_n
)_{n\in\mathbb{N}}\subset V^*$ with  $\xi_n\rightharpoonup \xi\in V^*$
we have 
\begin{align}
\label{eq:Psi*liminf}
\Psi^*_u(\xi)\leq  \liminf_{n\rightarrow\infty}\Psi^*_{u_n}(\xi_n).
\end{align}
\end{itemize}
\end{rem}

Finally, we make the following assumptions on the non-variational perturbation $B$.
\begin{enumerate}[label= (\thesection.B\alph*), leftmargin=3.2em]

\item \label{eq:B.1} \textbf{Continuity.} The map $(t,u) \mapsto
  B(t,u):[0,T]\times V \rightarrow V^*$ is continuous on sublevels of
  $G$, i.e.\ $(t_n,u_n)\to (t,u)$ in $[0,T]\ti V$ and
  $\sup_{n\in\mathbb{N}}G(u_n)\leq R$ implies $B(t_n,u_n)\to B(t,u)$
  in $V^*$.
  
\item \label{eq:B.2} \textbf{Control of $B$ by the energy.} There
  exist  $\beta>0$ and $c\in(0,1)$ such that
  \begin{align*}
    c\,\Psi^*_u\left(\frac1c B(t,u)\right)\leq \beta \big(1+\calE_t(u)\big)
    \quad \text{ for all } u\in D, \, t\in[0,T].
  \end{align*}
\end{enumerate}


\begin{rem}\label{re:Assump.B} 
  We note that Assumption \ref{eq:B.1} ensures that the
  \textsc{Nemytskij} operator associated to $B$ maps strongly
  measurable functions contained in sublevels of $G$ into strongly measurable functions, i.e., for
  all strongly measurable functions $u$ with $\sup_{t\in[0,T]}G(u(t))\leq R$, the map $t\mapsto B(t,u(t))$
  is strongly measurable.
\end{rem}

\subsection{Statement of the existence result}
\label{su:StateExistRes}

Before we state the main result, we say that $u\in \AC([0,T];V)$ is a
solution to \eqref{eq:I.1} with the initial datum $u_0\in D$ if $u$ satisfies the differential inclusion \eqref{eq:I.1} with $u(0)=u_0$.

\begin{thm}[Main existence result for $\PG=(V,\calE,\Psi,B)$]
  \label{th:MainExist}  Let the perturbed
  gradient system $(V,\calE,\Psi,B)$ satisfy the Assumptions
  \textnormal{(2.E), (2.$\Uppsi$)}, and 
  \textnormal{(2.B)}.  Then for every $u_0\in D$ there exists a solution
  $u\in \AC([0,T];V)$ to \eqref{eq:I.1} with $u(0)=u_0$ and an
  integrable function $\xi \in \rmL^1(0,T;V)$ with $\xi(t)\in \partial
  \calE_t(u(t))$ for a.a. $t\in(0,T)$ such that the following
  energy-dissipation balance holds:
\begin{align}
\label{eq:EDB}
\begin{split}
  & \calE_t(u(t))+ \int_s^t \left( \Psi_{u(r)}(u'(r)) +
    \Psi_{u(r)}^*\big(B(r,u(r))-\xi(r)\big) \right) \dd r 
\\ 
  &= \calE_s(u(s))+\int_s^t \partial_r \calE_r(u(r))\dd r +\int_s^t \langle
  B(r,u(r)),u'(r)\rangle \dd r \quad \text{for all } s,t\in [0,T].
\end{split}
\end{align}
\end{thm} 

 It is clear that every solution of \eqref{eq:EDB} is already a
solution for the perturbed gradient system $\PG=(V,\calE,\Psi,B)$,
since by the chain rule can and the \textsc{Legendre-Fenchel} theory we easily
recover \eqref{eq:I.1}, see e.g.\ \cite{AmGiSa05GFMS,RosSav06GFNC}. 

Our proof will be done by time discretization and solving variational
problems for each time interval $(t_n,t_{n+1}]$. To obtain a useful
discrete counterpart of the energy-dissipation balance proper we
employ \textsc{De Giorgi}'s variational interpolant, see
\cite[Lem.\,2.5]{Ambr95MM} or
\cite[Sec.\,4.2]{RosSav06GFNC}. We then follow the ideas in
\cite{MiRoSa13NADN}, but need to generalize to the case of a
nontrivial perturbation $B$, which only satisfies our mild assumptions
\ref{eq:B.1} and \ref{eq:B.2}. The proof will be completed in Section
\ref{su:Proof}.

\subsection{Estimates on the {\mdseries\scshape Moreau-Yosida} regularization}
\label{su:MoreauYosida}

In order to prove the existence result, we need to show some
properties of the $\Psi$-\textsc{Moreau-Yosida} regularization
\begin{align*}
\Phi_{r,t}(w;u):= \inf_{v\in V} \Upphi(r,t,u,w;v)
\end{align*} for $r>0, t\in [0,T)$ with $r+t\in [0,T]$ and $u\in D$ as
well as $w\in V^*$. Therefore, we have to ensure that the resolvent
set $J_{r,t}(w;u):= \argmin_{v\in V} \Upphi(r,t,u,w;v)$ is not empty.

\begin{lem}\label{le:Exist.Min} Let the perturbed
  gradient system $(V,\calE,\Psi,B)$ satisfy the Assumptions
  \textnormal{\ref{eq:cond.E.1}-\ref{eq:cond.E.2}} and
  \textnormal{\ref{eq:Psi.1}}. Then for all $r>0$, $t\in [0,T)$ with
  $t+r\leq T$, $u\in D $, and $w\in V^*$, the resolvent set
  $J_{r,t}(w;u)$ is nonempty.
\end{lem}
\begin{proof} Let $u\in D, w\in V^*$ and $r>0, t\in [0,T)$ with
  $r+t\in [0,T]$ be given.  First of all, we see with the
  \textsc{Fenchel-Young} inequality and with \eqref{eq:II.7} that
\begin{align}
\label{eq:II.11}
  \Upphi(r,t,u,w;v)&= r\Psi_{u}\left(\frac{v-u}{r}\right)+\calE_{t+r}(v)-\langle w,v\rangle \notag \\
  &\geq -r\Psi_u^*(w)+\calE_{t+r}(v)-\langle w,u\rangle\\
  &\geq -r\Psi_u^*(w) +S -\langle w,u\rangle \notag.
\end{align} 
This implies $\Phi_{r,t}(w;u)>-\infty$. On the other hand, we observe that
\begin{align}
\label{eq:II.12}
  \inf_{v\in V}\Big\lbrace
  r\Psi_{u}\left(\frac{v-u}{r}\right)+\calE_{t+r}(v)-\langle
  w,v\rangle \Big\rbrace \leq \calE_{t+r}(u)-\langle w,u\rangle,
\end{align} 
so that we also have $\Phi_{r,t}(w;u)<+\infty$.  Let now $(v_n)_{n \in
  \mathbb{N}}\subset V$ be a minimizing sequence for
$\Upphi(r,t,u,w;\cdot)$. From \eqref{eq:II.11}, we deduce with \eqref{eq:II.5} that
$(v_n)_{n \in \mathbb{N}}\subset V$ is contained in a sublevel set of
$\calE_t$. Thus, by Assumption \ref{eq:cond.E.2} and Remark \ref{re:Assump.E} $iv)$ there exists
a subsequence (not relabeled) which converges strongly in $V$ towards a
limit $ v\in V$. Together with the lower semicontinuity of the map
$v\mapsto \Upphi(r,t,u,w;v)$, we have
\begin{align*}
  \Upphi(r,t,u,w;v)\leq \liminf_{n\rightarrow
    \infty}\Upphi(r,t,u,w;v_n)=\inf_{\tilde{v}\in V}
  \Upphi(r,t,u,w;\tilde{v})
\end{align*} and therefore $v\in J_{r,t}(w;u)\neq \emptyset$ from what $v\in D$ follows.
\end{proof}
Lemma \ref{le:Exist.Min} is important for justifying the existence of a sequence of approximate values
$(U_\tau^n)_{n=1}^N\subset D$ that complies with
\begin{align}
\label{eq:II.13}
  U_\tau^n \in J_{\tau,t_{n-1}}(B(t_{n-1},U_\tau^{n-1}),U_\tau^{n-1})
  \quad \text{for all } n=1, \cdots, N,
\end{align} 
in order to construct discrete solutions of \eqref{eq:II.2}, where
$U_\tau^0:=u_0$ and the time $t\in [0,T)$ as well as the time step
$\tau\in (0,T-t)$ are fixed. 

The following lemma is crucial in order to proof the existence result
and in particular to derive a priori estimates for the interpolation
functions we define later on. The result is an adaptation to the
case $w\neq 0$ of \cite[Lem.\,4.2]{RosSav06GFNC} and
\cite[Lem.\,6.1]{MiRoSa13NADN}.

\begin{lem}\label{le:Main.Lem} Let the perturbed
  gradient system $(V,\calE,\Psi,B)$ satisfy the Assumptions
  \textnormal{(2.E), (2.$\Uppsi$)}, and 
  \textnormal{(2.B)}. Then for every $\te\in [0,T), u\in D$ and 
  $w\in V^*$ there exists a measurable selection $r\mapsto u_r:
  (0,T-\te)\rightarrow J_{r,\te}(w;u)$ such that
\begin{align}
\label{eq:II.14}
w\in \partial \Psi_u \left(\frac{u_r-u}{r}\right) + \partial \calE_{\te+r}(u) 
\end{align} and there exists a constant $\widetilde{C}>0$ such that
\begin{align}
\label{eq:II.15}
G(u_r)\leq \widetilde{C} ( G(u)+r\Psi_u^*( w))\quad \text{for all } r\in(0,T-\te)
\end{align} Furthermore, there holds
\begin{align}
\label{eq:II.16}
\lim_{r\rightarrow 0}\sup_{u_r\in J_{r,\te}(w;u)}\Vert u_r-u\Vert=0 \quad \text{and}\quad \lim_{r\rightarrow 0} \Phi_{r,\te}(w;u)= \calE_\te(u)-\langle w,u\rangle
\end{align}
 for all $\te\in [0,T), u\in D$ and $w\in V^*$. Finally the map $r\mapsto \Phi_{r,\te}(w;u)$ is almost everywhere differentiable in $(0,T-\te)$ and for every $r_0\in (0,T-\te)$ and every measurable selection $r\mapsto u_r: (0,r_0)\rightarrow J_{r,\te}(w;u)$ there exists a measurable selection $r\mapsto \xi_r: (0,T-\te) \rightarrow  \partial \calE_{\te+r}(u)$ with $w-\xi_r\in \partial \Psi_u \left(\frac{u_r-u}{r}\right)$ such that
\begin{align}
\label{eq:II.17}
\begin{split}
E_{\te+r_0}(u_{r_0}) + r_0\Psi_u\left(\frac{u_{r_0}-u}{r_0}\right)
+&\int_0^{r_0} \Psi_u^*(w-\xi_r)\dd r \\ 
&\leq \calE_t(u)+\int_0^{r_0}\partial_r \calE_{\te+r}(u_r)\dd r +\langle w,u_{r_0}-u\rangle.
\end{split}
\end{align} 
\end{lem}
\begin{proof}
  Let $\te\in [0,T), u\in D$ and $w\in V^*$ be given. The
  non-emptiness of the resolvent set $J_{r,\te}(w;u)$ for all
  $r\in(0,T-\te)$ is guaranteed by Lemma \ref{le:Main.Lem}. The existence of a
  measurable selection $r \mapsto u_r: (0,T-\te)\rightarrow
  J_{r,\te}(w;u)$ is provided by \textsc{Castaing} and
  \textsc{Valadier} \cite[Cor.\,III.3, Prop.\,III.4, Thm.\,III.6,
  pp.\,63]{CasVal77CAMM}. The inclusion \eqref{eq:II.14} follows then by the
  variational sum rule \eqref{eq:II.8}. Further, we obtain from \eqref{eq:II.11} for
  $v=u_r, r\in(0,T-\te)$ and \eqref{eq:II.12} the inequality
\begin{align*}
\calE_{\te+r}(u_r)\leq \calE_{\te+r}(u)+r\Psi_u^*(w),
\end{align*} so that together with the estimate \eqref{eq:II.6} it follows the inequality \eqref{eq:II.15} with $\widetilde{C}=C_1$, where $C_1>0$ is the constant in \eqref{eq:II.6}. In order to show the convergences in \eqref{eq:II.16}, we note that Assumption \ref{eq:Psi.2} implies: For all $R>0$ and $\gamma>0$, there exists $K>0$ such that 
\begin{align*}
\Psi_u(v)\geq \gamma\Vert v\Vert
\end{align*} for all $u\in D$ with $G(u)\leq R$ and all $v\in V$ with $\Vert v\Vert\leq K$.
 Based on this fact, we infer
\begin{align}
\label{eq:II.18}
\gamma \left \Vert \frac{u_r-u}{r} \right\Vert &\leq  \Psi_u\left(\frac{u_r-u}{r}\right)+\gamma K \quad \text{ for every }r>0.
\end{align} Together with \eqref{eq:II.7}, \eqref{eq:II.11} and \eqref{eq:II.12}, we obtain
\begin{align*}
\gamma \Vert u_r-u \Vert &\leq \langle w, u_r-u\rangle +\calE_{\te+r}(u)-\calE_{\te+r}(u_r) +r\gamma K\\
&\leq \Vert w \Vert \Vert u_r-u\Vert +\calE_{\te+r}(u)-S+r\gamma K.
\end{align*} 
This implies the estimate
\begin{align*}
(\gamma-\Vert w\Vert_*) \Vert u_r{-} u \Vert &\leq
\calE_{\te+r}(u)-S+r\gamma K
\leq \ee^{CT}\calE_0(u)-S+r\gamma K 
\end{align*} for all $\gamma>0, r\in (0,T-\te)$ and $u_r\in J_{r,\te}(w;u)$, where we used again \eqref{eq:II.5}. By taking the supremum over all $u_r\in J_{r,\te}(w;u)$ and taking the limes superior as $r\rightarrow 0$, we finally obtain
\begin{align*}
(\gamma-\Vert w\Vert_*) \limsup_{r\rightarrow 0} \sup_{u_r\in J_{r,\te}(w;u)} \Vert u_r-u \Vert \leq  \ee^{CT} \calE_0(u)-S\quad \text{for every }\gamma>\Vert w\Vert_*.
\end{align*} By choosing $\gamma>0$ sufficiently large, we conclude 
\begin{align*}
\limsup_{r\rightarrow 0} \sup_{u_r\in J_{r,\te}(w;u)} \Vert u_r-u \Vert=0,
\end{align*} 
which shows the first convergence in \eqref{eq:II.16}. We now use 
the lower semicontinuity and the time continuity of the energy
functional, the estimate
\begin{align*}
&\calE_{\te+r}(u_r)-\langle w,u_r\rangle \leq \Phi_{r,\te}(w;u)\\
&\quad =  r\Psi_{u}\left(\frac{u_r-u}{r}\right)+\calE_{\te+r}(u_r)-\langle
w,u_r\rangle 
\leq \calE_{\te+r}(u)-\langle w,u\rangle,
\end{align*} 
and the fact that $\liminf_{r\rightarrow
  0}\calE_{\te+r}(u_r)=\liminf_{r\rightarrow 0}\calE_{\te}(u_r)$,
which follows from \eqref{eq:II.5}. Hence, the second convergence in
\eqref{eq:II.16} follows from the estimate 
\begin{align*}
&\hspace*{-1em}\calE_\te(u)-\langle w, u\rangle\leq \liminf_{r\rightarrow 0} \left( \calE_{\te+r}(u_r)-\langle w, u_r\rangle\right) \\
&\leq \liminf_{r\rightarrow 0} \Phi_{r,t}(w;u) 
\ \leq \ \limsup_{r\rightarrow 0}\Phi_{r,\te}(w;u)\\
&\leq \limsup_{r\rightarrow 0} \left( \calE_{\te+r}(u)-\langle
  w,u\rangle \right)
\ = \ \calE_\te(u)-\langle w,u\rangle.
\end{align*} 

In order to show the last assertion of this lemma, let $u_{r_i}\in J_{r,\te}(w;u), i=1,2,$ with $0<r_1<r_2<T-\te$. Then we have
\begin{align}
&\Phi_{r_2,\te}(w;u)-\Phi_{r_1,\te}(w;u) - \left( \calE_{\te+r_2}(u_{r_1}) - \calE_{\te+r_1}(u_{r_1})\right) \notag \\
&\leq r_2\Psi_u\left(\frac{u_{r_1}-u}{r_2}\right)-r_1\Psi_u\left(\frac{u_{r_1}-u}{r_1}\right)\notag \\
&= (r_2-r_1)\Psi_u \left(\frac{u_{r_1}-u}{r_2} \right)+r_1 \left( \Psi_u \left(\frac{u_{r_1}-u}{r_2}\right)-\Psi_u \left(\frac{u_{r_1}-u}{r_1} \right) \right) \notag\\
\label{eq:II.19}
&\leq (r_2-r_1)\left( \Psi_u \left(\frac{u_{r_1}-u}{r_2} \right)-\left\langle w_2^1,\frac{u_{r_1}-u}{r_2} \right\rangle \right)\\
\label{eq:II.20}
&=-(r_2-r_1)\Psi_u^*(w_2^1)\leq 0,
\end{align} 
where we used in \eqref{eq:II.19} the fact from Remark
\ref{re:Assump.Psi} $i)$ which states $w_2^1\in \partial \Psi_u \left
  (\frac{u_{r_1}-u}{r_2}\right)\neq \emptyset$, in \eqref{eq:II.20}
the statement of Lemma \ref{le:Leg.Fen} and the last inequality the
fact that by the \textsc{Fenchel-Young} inequality we have
$\Psi_u^*(w)\geq 0$ for all $w\in V^*$. Further, we deduce with the
aid of \ref{eq:cond.E.3}, \eqref{eq:II.6} and the already proven
inequality \eqref{eq:II.15} that
\begin{align}
\label{eq:II.21}
&\Phi_{r_2,\te}(w;u)\leq \Phi_{r_1,\te}(w;u) + \left( \calE_{\te+r_2}(u_{r_1}) - \calE_{\te+r_1}(u_{r_1})\right) \notag \\
&= \Phi_{r_1,\te}(w;u) - \int_{r_1}^{r_2} \partial_r \calE_{\te+r}(u_{r_1})\dd r  \notag\\
&\leq \Phi_{r_1,\te}(w;u) +(r_2-r_1) CC_1 G(u_{r_1}) \notag\\
&\leq \Phi_{r_1,\te}(w;u) +(r_2-r_1) CC_1 ( G(u)+r_1\Psi^*(w)) \notag\\
&\leq \Phi_{r_1,\te}(w;u) +(r_2-r_1) CC_1 ( G(u)+T\Psi^*(w)).
\end{align} We conclude that the map $r\mapsto \Phi_{r,\te}(w;u) -r
CC_1 ( G(u)+T\Psi^*(w))$ is non-increasing on $(0,T-\te)$ and
therefore as a real-valued function almost everywhere
differentiable. Since the map $r \mapsto \Phi_{r,\te}(w;u)$ is a
linear perturbation of a monotone function, it is also almost
everywhere differentiable in $(0,T-\te)$. Thus there exists a
negligible set $\mathscr{N}\subset (0,T-\te)$, such that the map
$r\mapsto \Phi_{r,\te}(w;u)$ is differentiable on $(0,T-\te)\backslash
\mathscr{N}$. We remark that the negligible set depends on $u$ and
$w$, that is $\mathscr{N}=\mathscr{N}_{u,w}$. Now, to conclude, we
want to use the inequality \eqref{eq:II.20}. For this let
$r\in(0,T-\te)\backslash \mathscr{N}$ be fixed. Additionally let
$(h_n)_{n\in \mathbb{N}}\in \mathbb{R}^{>0}$ be a sequence which
converges from above towards zero and whose elements are sufficiently
small. Let also the sequence $(w^r_n)_{n\in \mathbb{N}}\subset V^*$ be
given by $w^r_n\in \partial \Psi_u \left
  (\frac{u_{r}-u}{r+h_n}\right)$ for all $n\in \mathbb{N}$. The
boundedness of the operator $\partial\Psi_u$ according to Remark
\ref{re:Assump.Psi} $i)$ implies that the sequence $(w^r_n)_{n\in
  \mathbb{N}}\subset V^*$ is bounded in $V^*$. Thus there exists a
subsequence, labeled as before, and an element $w_r\in V^* $ such that
$w^r_n\rightharpoonup w_r$ weakly in $V^*$. From the strong-weak
closedness of the graph of $\partial \Psi_u$ in $V\times V^*$ it
follows $w_r \in \partial \Psi_u \left
  (\frac{u_{r}-u}{r}\right)$. Since the conjugate $\Psi_u^*$ is convex
and lower semicontinuous, it is also weakly lower semicontinuous. Then
we find with Lemma \ref{le:Leg.Fen} and the continuity of $\Psi_u$
that
\begin{align*}
  \Psi_u^*(w_r)&\leq \liminf_{n\rightarrow \infty} \Psi_u^*(w^r_n)
  \leq\limsup_{n\rightarrow \infty} \Psi_u^*(w^r_n)\\
  &=  \limsup_{n\rightarrow \infty} \left( \left\langle w_n^r,\frac{u_{r}-u}{r+h_n} \right\rangle- \Psi_u \left(\frac{u_{r}-u}{r+h_n} \right) \right)\\
  &= \left\langle w_r,\frac{u_{r}-u}{r} \right\rangle- \Psi_u
  \left(\frac{u_{r}-u}{r} \right)=\Psi_u^*(w_r)
\end{align*} and thus $\lim_{n\rightarrow \infty}
\Psi_u^*(w^r_n)=\Psi_u^*(w_r)$. Due to the inclusion \eqref{eq:II.14} there
exist $\xi_r\in \partial \calE_{\te+r}(u)$ such that $w-\xi_r \in \partial
\Psi_u \left(\frac{u_{r}-u}{r} \right)$. By \textsc{Aubin} and
\textsc{Frankowska} \cite[Thm.\,8.2.9, p.\,315]{AubFra90SVA}, the
selection $r\mapsto \xi_r: (0,T-\te)\rightarrow \partial \calE_{\te+r}(u)$
can be chosen to be measurable. Further, from Assumption \ref{eq:Psi.1} we
get $\Psi_u^*(w_r)= \Psi_u^*(w-\xi_r)$. By the differentiability of
the map $r \mapsto\Phi_{r,\te}(w;u)$ in $r$, we obtain with \eqref{eq:II.20}
\begin{align}
 & \frac{\rmd}{\rmd r} \Phi_{r,\te}(w;u)\vert_{r=r}+ \Psi_u^*(w-\xi_r)
   \ = \ \lim_{n\rightarrow \infty}\left( \frac{\Phi_{r+h_n,\te}(w;u)- 
  \Phi_{r,\te}(w;u)}{h_n} + \Psi_u^*(w_n^r) \right)\notag \\
\label{eq:II.22}
  &\leq \liminf_{n\rightarrow \infty}	\left( 
   \frac{ E_{\te+r+h_n}(u_{r}) -
    \calE_{\te+r}(u_{r})}{h_n}\right)
   \ = \ \partial_t E_{\te +r}(u_r) \quad \text{for a.a. }r\in(0,T{-}\te),
\end{align} 
where we also used the fact that the map $t \mapsto \calE_t$ is
differentiable. The claim finally follows by integrating \eqref{eq:II.22} from
$r=0$ to $r=r_0$ and by using \eqref{eq:II.16}.
\end{proof}

\subsection{Time discretization and discrete energy-dissipation estimate}
\label{su:TimeDiscret}

With the help of the preceding lemma, we derive in the forthcoming
result a priori estimates for the approximate solutions, more
precisely for both the piecewise constant interpolation functions
$\overline{U}_\tau $ and $\underline{U}_\tau$, and for the piecewise
linear interpolation function $\widehat{U}_\tau$ as well as for the
so-called \textsc{De Giorgi} interpolation function
$\widetilde{U}_\tau$. In order to define the interpolation functions,
let the initial value $u_0\in D$ and the time step $\tau>0$ be
fixed. Further let $(U_\tau^n)_{n=1}^N \subset D$ be the sequence of
approximate values, which are defined by the variational approximation
scheme
\begin{align}
\label{eq:II.23}
\begin{cases} \quad U_\tau^0=u_0, \\ \quad 
U_\tau^n\in J_\tau(B(t_{n-1},U_\tau^{n-1});U_\tau^{n-1})), \quad n=1,2,\dots,N.
\end{cases} 
\end{align}
The piecewise constant and linear interpolation functions we define by
\begin{align}
\label{eq:Approx.tau}
  &\overline{U}_\tau(0)=\underline{U}_\tau(0)=\widehat{U}_\tau(0):=U_\tau^0 \text{ and } \notag \\
  &\underline{U}_\tau(t):=U^{n-1}, \quad \widehat{U}_\tau(t):=\frac{t_n-t}{\tau}U_\tau^{n-1}+\frac{t-t_{n-1}}{\tau}U_\tau^{n} \quad \text{for } t\in[t_{n-1},t_n),\notag\\
  &\overline{U}_\tau(t):=U_\tau^n \quad \text{ for } t\in(t_{n-1},t_n]
  \quad \text{and all } n=1,\dots,N.
\end{align} 

Furthermore, we define by the approximation scheme  
\begin{align}
\label{eq:II.25}
\begin{cases}\quad
  \widetilde{U}_\tau(0):=U_\tau^0,  \\ \quad
  \widetilde{U}_\tau(t) \in J_r(B(t_{n-1},U_\tau^{n-1});U_\tau^{n-1}))
  \quad \text{ for }t =t_{n-1}+r \in (t_{n-1},t_n],
\end{cases} 
\end{align} $n=1,2,\dots,N,$ the \textsc{De Giorgi} interpolation
$\widetilde{U}_\tau$. We note that we can assume the measurability of
the function $\widetilde{U}_\tau$ since by Lemma \ref{le:Main.Lem} there always
exists a measurable selection of the \textsc{De Giorgi}
interpolation. Due to the fact that for all $t\in I_\tau $ the approximation scheme \eqref{eq:II.25}
yields the usual scheme in \eqref{eq:II.23}, we can assume without loss of
generality that all interpolation functions coincide on the nodes
$t_n$, i.e.,
\begin{align*}
\widetilde{U}_\tau(t_n)=\overline{U}_\tau(t_n)=\underline{U}_\tau(t_n)=\widehat{U}_\tau(t_n)=U_\tau^n \quad \text{for all } n=1,\cdots,N.
\end{align*} Moreover, we denote by $\widetilde{\xi_\tau}$ the
interpolation function obtained from Remark \ref{re:Assump.E}  $iii)$ with the
variational sum rule by choosing $t=t_{n-1},
u_0=\widetilde{U}_\tau(t), u=U_\tau^{n-1}$ and
$w=B(t_{n-1},U_\tau^{n-1})$, and which satisfies
\begin{align}
\label{eq:II.26}
  \widetilde{\xi_\tau}(t)\in \partial
  \calE_{t_{n-1}+r}(\widetilde{U}_\tau(t)) \quad \text{ for }
  t=t_{n-1}+r\in (t_{n-1},t_n],
\end{align} and 
\begin{align}
\label{eq:II.27}
  B(t_{n-1},U_\tau^{n-1})-\widetilde{\xi_\tau}(t) \in \partial
  \Psi_{U_\tau^{n-1}}\left (\frac{\widetilde{U_\tau}(t)-U_\tau^{n-1}
    }{t-t_{n-1}}\right) \text{ for } t=t_{n-1}+r\in (t_{n-1},t_n]
\end{align} for all $n=1,\dots,N$. The measurability of the function
$\widetilde{\xi_\tau}:(0,T)\rightarrow V^*$ again follows from Lemma \ref{le:Main.Lem}.

 For notational convenience, we also introduce the piecewise
constant interpolation functions $\bar{\mathbf{t}}_\tau : [0,T]
\rightarrow [0,T]$ and $\teu_\tau:[0,T]\rightarrow [0,T]$ given by
\begin{align*}
&\teo_{{\tau}}(0):= 0\,\,\,\, \text{ and } \, \teo_{{\tau}}(t):=
t_n\quad \, \, \, \, \, \text{ for } t\in (t_{n-1},t_n], \quad n=1,\dots,N,\\
&\teu_{{\tau}}(T):= T \, \text{ and } \, \teu_{{\tau}}(t):=
t_{n-1} \quad \text{ for } t\in [t_{n-1},t_n), \quad n=1,\dots,N.
\end{align*} Obviously, there holds $\teo_\tau(t)\rightarrow t$ and
$\teu_\tau(t)\rightarrow t$ as $\tau\rightarrow 0$.

We are now in the position to show a priori estimates for the
approximate solutions.

\begin{lem}\label{le:DUEE} Let the perturbed
  gradient system $(V,\calE,\Psi,B)$ satisfy the Assumptions
  \textnormal{(2.E), (2.$\Uppsi$)}, and 
  \textnormal{(2.B)}.  Furthermore, let $\widetilde{U}_\tau,
  \overline{U}_\tau,\underline{U}_\tau, \widehat{U}_\tau$ and
  $\widetilde{\xi}_\tau$ be the interpolation functions defined in
  \eqref{eq:Approx.tau}-\eqref{eq:II.26} associated to a fixed initial datum $u_0\in D$ and a
  step size $\tau>0$.\\ Then, the discrete upper energy estimate
\begin{align}
\label{eq: DUEE}
  \calE_{\teo_\tau(t)}(\overline{U}_\tau(t)) + \int_{\teo_\tau(s)}^{\teo_\tau(t)}\left(
    \Psi_{\underline{U}_\tau(r)}\left( \widehat{U}'_\tau(r)\right) +
    \Psi^*_{\underline{U}_\tau(r)}\left(
      B(\teu_\tau(r),\underline{U}_\tau(r))-
      \widetilde{\xi}_\tau(r)\right) \right) \dd r
   \notag \\ 
  \leq
  \calE_{\teo_\tau(s)}(\overline{U}_\tau(s))+\int_{\teo_\tau(s)}^{\teo_\tau(t)} \partial_r
  \calE_r(\widetilde{U}_\tau(r))\dd r+\int_{\teo_\tau(s)}^{\teo_\tau(t)}
  \langle B(\teu_\tau(r),\underline{U}_\tau(r)), \widehat{U}'_\tau
  (r) \rangle \dd r
\end{align} holds for all $0\leq s< t\leq T$. Moreover, there exist positive constants $M,\tau^*>0$ such that the estimates 
\begin{align}
\label{eq:II.29}
\sup_{t\in (0,T)} \calE_t(\overline{(U}_\tau(t)) \leq M,\quad \sup_{t\in (0,T)} \calE_t(\widetilde{(U}_\tau(t))\leq M,\quad \sup_{t\in (0,T)}\vert \partial_t \calE_t(\widetilde{(U}_\tau(t))\vert \leq M \\ 
\label{eq:II.30}
\int_0^T \left( \Psi_{\underline{U}_\tau(r)}\left( \widehat{U}'_\tau(r)\right) + \Psi^*_{\underline{U}_\tau(r)}\left( B(\teu_\tau(r),\underline{U}_\tau(r))- \widetilde{\xi}_\tau(r)\right) \right) \dd r\leq M
\end{align} hold for all $0<\tau\leq \tau^*$. Besides, the families $(\widehat{U}'_\tau)_{0<\tau\leq \tau^*}\subset \rmL^1(0,T;V)$ as well as $(B(\teu_\tau,\underline{U}_\tau))_{0<\tau\leq \tau^*}\subset \rmL^1(0,T;V^*)$ and  $(\widetilde{\xi}_\tau)_{0<\tau\leq \tau^*} \subset \rmL^1(0,T;V^*)$ are integrable uniformly with respect to $\tau$ in the respective spaces. Finally, there holds
\begin{align}
\label{eq:II.31}
\Vert \underline{U}_\tau-\overline{U}_\tau\Vert_{\infty}+\Vert \widehat{U}_\tau-\overline{U}_\tau\Vert_{\infty}
+\Vert \widetilde{U}_\tau-\underline{U}_\tau\Vert_{\infty} \rightarrow 0
\end{align} as $\tau \rightarrow 0$.
\end{lem}
\begin{proof}
In order to show the discrete upper energy estimate \eqref{eq: DUEE}, it is sufficient to restrict ourselves to the case $s=t_{n-1}$ and $t=t_n$ for $n\in n=1,\dots, N$. The general case follows by summing up the particular inequalities on the subintervals. But this case follows from \eqref{eq:II.17} in Lemma \ref{le:Main.Lem} by choosing $\te=t_{n-1},u=U_\tau^{n-1}, r_0=t-t_{n-1}, u_{r_0}=\widetilde{U}_\tau(t),  u_{r}=\widetilde{U}_\tau(t_{n-1}+r)$ and $\xi_r=\widetilde{\xi}_\tau(t_{n-1}+r)$, where we chose $t\in(t_{n-1},t_n]$ to be fixed. Then, we find 
\begin{align}
\label{eq:II.32}
(t-t_{n-1})\Psi_{U_\tau^{n-1}}\left( \frac{\widetilde{U}_\tau(t)-U_\tau^{n-1}}{t-t_{n-1}} \right)+\int_{t_{n-1}}^t \Psi_{U_\tau^{n-1}} \left( B(t_{n-1},U_\tau^{n-1})- \widetilde{\xi}_\tau(r)\right)\dd r+\calE_t(\widetilde{U}_\tau(t)) \notag \\
\leq \calE_{t_{n-1}}(U_\tau^{n-1})+\int_{t_{n-1}}^{t} \partial_r \calE_r(\widetilde{U}_\tau(r))\dd r+ \langle B(t_{n-1},U_\tau^{n-1}), U_\tau^n-U_\tau^{n-1} \rangle.
\end{align} By choosing $t=t_n$, we obtain
\begin{align}
\label{eq:II.33}
\int_{t_{n-1}}^{t_n}\left( \Psi_{\underline{U}_\tau(r)}\left( \widehat{U}'_\tau(r)\right) + \Psi^*_{\underline{U}_\tau(r)}\left( B(t_{n-1},\underline{U}_\tau(r))- \widetilde{\xi}_\tau(r)\right) \right) \dd r +\calE_{t_n}(\overline{U}_\tau(t_n))\notag \\
\leq \calE_{t_{n-1}}(\underline{U}_\tau(t_{n-1}))+\int_{t_{n-1}}^{t_n} \partial_r \calE_r(\widetilde{U}_\tau(r))\dd r+\int_{t_{n-1}}^{t_n}\langle B(t_{n-1},\underline{U}_\tau(r)), \widehat{U}'_\tau (r) \rangle \dd r
\end{align} for all $n=1,\cdots,N$, which yields the discrete upper energy estimate. Further, we notice that from Assumption \ref{eq:B.2}, we obtain the estimation
\begin{align}
\label{eq:II.34}
&\int_{t_{n-1}}^{t_n}\langle B(t_{n-1},\underline{U}_\tau(r)), \widehat{U}'_\tau (r) \rangle \dd r \notag\\
 &\leq c \int_{t_{n-1}}^{t_n}\Psi_{\underline{U}_\tau(r)}\left( \widehat{U}'_\tau(r)\right) \dd r+c \int_{t_{n-1}}^{t_n} \Psi^*_{\underline{U}_\tau(r)}\left( \frac{B(t_{n-1},\underline{U}_\tau(r))}{c}\right) \dd r \notag \\
&\leq c \int_{t_{n-1}}^{t_n}\Psi_{\underline{U}_\tau(r)}\left( \widehat{U}'_\tau(r)\right) \dd r+\tau \beta(1+  \calE_{t_{n-1}}(U_\tau^{n-1}))\notag \\
&\leq c \int_{t_{n-1}}^{t_n}\Psi_{\underline{U}_\tau(r)}\left( \widehat{U}'_\tau(r)\right) \dd r+ \tau \beta  (1+G(U_\tau^{n-1})),
\end{align} where we used also the \textsc{Fenchel-Young} inequality. Since $c\in(0,1)$, inequality \eqref{eq:II.33} and \eqref{eq:II.34} together yield the estimation
\begin{align}
\calE_{t_n}(U_\tau^n) 
&\leq \calE_{t_{n-1}}(U_\tau^{n-1})+\int_{t_{n-1}}^{t_n} \partial_r \calE_r(\widetilde{U}_\tau(r))\dd r+ \tau \beta(1+ G(U_\tau^{n-1}))\notag\\
&\leq  \calE_{t_{n-1}}(U_\tau^{n-1})+ \tau \beta 
 (1+G(U_\tau^{n-1}))+C\widetilde{C}\int_{t_{n-1}}^{t_n} G(U_\tau^{n-1})\dd r \notag\\
 \label{eq:II.35}
&+  \int_{t_{n-1}}^{t_n}(r-t_{n-1})\Psi^*_{U_\tau^{n-1}}(B(t_{n-1},U_\tau^{n-1}))\dd r \\ 
&\leq  \calE_{t_{n-1}}(U_\tau^{n-1})+ \tau \beta 
 (1+G(U_\tau^{n-1}))+C\widetilde{C}\int_{t_{n-1}}^{t_n} G(U_\tau^{n-1})\dd r \notag\\
 \label{eq:II.36}
&+  \int_{t_{n-1}}^{t_n} c \tau\Psi^*_{\underline{U}_\tau(r)}\left( \frac{B(t_{n-1},\underline{U}_\tau(r))}{c}\right) \dd r  \\
&\leq  \calE_{t_{n-1}}(U_\tau^{n-1})+ \tau \beta(1+  G(U_\tau^{n-1}))+C\widetilde{C}\tau G(U_\tau^{n-1})\notag \\
&+  \tau \beta (1+G(U_\tau^{n-1}))\notag \\
\label{eq:II.37}
&= \calE_{t_{n-1}}(U_\tau^{n-1})+\tau (2\beta  +C\widetilde{C}) G(U_\tau^{n-1})+2\tau\beta
\end{align} for all $n=1,\dots, N$ and $0<\tau\leq1 $, where we used in \eqref{eq:II.35} the inequality 
\begin{align*}
G(\widetilde{U}_\tau(t))\leq \widetilde{C}(G(U_\tau^{n-1})+(t-t_{n-1})\Psi_{U_\tau^{n-1}}^*(B(t_{n-1},U_\tau^{n-1}))), \quad t\in(t_{n-1},t_n],
\end{align*}
 from Lemma \ref{le:Main.Lem} and in \eqref{eq:II.36} the fact that the map $r\mapsto  r\Psi^*_u\left(\frac{\xi}{r}\right)$ is non-decreasing on $(0,+\infty)$ for every $\xi\in V^*$. Defining $A:=(2\beta  +C\widetilde{C})$ and summing up the inequalities \eqref{eq:II.37},  we obtain
\begin{align}
\label{eq:II.38}
G(U_\tau^n)\leq C_1 \calE_{t_n}(U_\tau^n) \leq C_1 \calE_0(u_0)+2C_1T\beta+ \tau C_1A
\sum_{k=1}^n G(U_\tau^{k-1}) 
\end{align}  for all $n=1,\dots, N$ and $0<\tau\leq1 $. Then, applying the discrete version of the \textsc{Gronwall} Lemma to \eqref{eq:II.38} yields the uniform boundedness of $G(U_\tau^n)$ for all $n= 1,\dots, N$ and $0<\tau<\min\lbrace 1, 1/(2C_1A)\rbrace=:\tau^*$, from what we deduce
\begin{align}
\label{eq:II.39}
\sup_{t\in(0,T)} \calE_t(\overline{U}_\tau(t)) \leq C_1\quad \text{for all } 0<\tau<\tau^*
\end{align} for a positive constant $C_1>0$ independent from $\tau$. Taking into account the inequality \eqref{eq:II.37} and the Assumptions \ref{eq:B.2} and \ref{eq:cond.E.3}, we also obtain the last two inequalities in \eqref{eq:II.29}. By employing \eqref{eq:II.33} and \eqref{eq:II.34}, and arguing as before, we also get \eqref{eq:II.30}. The constant $M$ can be chosen by the sum of all constants obtained from the shown inequalities of this lemma. Further, the uniform integrability of $(\widehat{U}'_\tau)_{0<\tau\leq \tau^*}$ as well as $(B(\teu_\tau,\underline{U}_\tau))_{0<\tau\leq \tau^*}$ and  $(\widetilde{\xi}_\tau)_{0<\tau\leq \tau^*}$ in $\rmL^1(0,T;V)$ and $\rmL^1(0,T;V^*)$, respectively, 
follows from the superlinear growth of $\Psi_u$ and $\Psi^*_u$ (Assumption \ref{eq:Psi.2}), inequality \eqref{eq:II.30} and the growth condition \ref{eq:B.2}. To clarify this, let $\varepsilon>0$ and $\widetilde{M}:=\max\lbrace \beta(1+M),M\rbrace$ be given, where $M$ is the constant obtained from the boundedness in \eqref{eq:II.29} and \eqref{eq:II.30}. Then, by Assumption \ref{eq:Psi.2} there exists for $M$ and  $\widetilde{M}/\varepsilon$ positive numbers $K_1,K_2$, such that 
\begin{align}
\label{eq:II.40}
\Psi_u(v)\geq \frac{\widetilde{M}}{\varepsilon}\Vert v\Vert \quad \text{ and }\quad \Psi_{u}^*(\eta)\geq \frac{\widetilde{M}}{\varepsilon}\Vert \eta \Vert_* 
\end{align} for all $v\in V$ with $\Vert v\Vert\geq K_1$, all $\eta \in V^*$ with $\Vert \eta \Vert_*\geq K_2$ and all $u\in D$ with $G(u)\leq M$, 
For notational convenience, we define $f_\tau: [0,T]\rightarrow V$, $g_\tau: [0,T]\rightarrow V^*$ and $h_\tau: [0,T]\rightarrow V^*$ by $f_\tau(t):= \widehat{U}'_\tau(t)$, $g_\tau(t):= B(\teu_\tau(t),\underline{U}_\tau(t))$ and $h_\tau(t):= (B(\teu_\tau(t),\underline{U}_\tau(t))-\widetilde{\xi}_\tau(t))$ for all $t\in [0,T]$. Then, by \eqref{eq:II.40}, \eqref{eq:II.29} and \eqref{eq:II.30} there hold
\begin{align*}
\int_{\lbrace t\in [0,T] : f_\tau(t)\geq K_1 \rbrace} \Vert f_\tau(t) \Vert \dd t \leq \frac{\varepsilon}{\widetilde{M}} \int_{\lbrace t\in [0,T] : f_\tau(t)\geq K_1 \rbrace} \Psi_{\underline{U}_\tau(t)}(f_\tau(t)) \dd t\leq \varepsilon \\
\int_{\lbrace t\in [0,T] : g_\tau(t)\geq K_2 \rbrace} \Vert g_\tau(t) \Vert_* \dd t \leq \frac{\varepsilon}{\widetilde{M}} \int_{\lbrace t\in [0,T] : g_\tau(t)\geq K_2 \rbrace} \Psi^*_{\underline{U}_\tau(t)}(g_\tau(t)) \dd t\leq \varepsilon\\
\int_{\lbrace t\in [0,T] : h_\tau(t)\geq K_2 \rbrace} \Vert h_\tau(t) \Vert_* \dd t \leq \frac{\varepsilon}{\widetilde{M}} \int_{\lbrace t\in [0,T] : h_\tau(t)\geq K_2 \rbrace} \Psi^*_{\underline{U}_\tau(t)}(h_\tau(t)) \dd t\leq \varepsilon
\end{align*} for all $0< \tau \leq \tau^*$, which yields the uniform integrability. Since the sum of two uniformly integrable functions is again uniformly integrable, it follows that $(\widetilde{\xi}_\tau)_{0<\tau\leq \tau^*}$ is also uniformly integrable in $\rmL^1(0,T;V^*)$ with respect to $\tau>0$. For the last assertion, we first notice that inequality \eqref{eq:II.32} considering \eqref{eq:II.29} and \eqref{eq:II.30} implies
\begin{align*}
\sup_{t\in [0,T]}(t-\teu_\tau(t))\Psi_{\underline{U}_\tau(t)}\left( \frac{\widetilde{U}_\tau(t)-\underline{U}_\tau(t)}{t-\teu_\tau(t)} \right)\leq C_2 .
\end{align*} for a constant $C_2>0$. Then, again Assumption \ref{eq:Psi.2} implies that for every $R>0$ and $\gamma>0$ there exists $K>0$ such that
\begin{align}
\label{eq:II.41}
\gamma \Vert \widetilde{U}_\tau(t)-\underline{U}_\tau(t) \Vert &\leq (t-\teu_\tau(t))\Psi_{\underline{U}_\tau(t)}\left( \frac{\widetilde{U}_\tau(t)-\underline{U}_\tau(t)}{t-\teu_\tau(t)} \right)+ (t-\teu_\tau(t))\gamma K \notag \\
&\leq M+\tau \gamma K \quad \text{ for all } t\in [0,T] \text{ and all } 0<\tau<\tau^*.
\end{align} 
Taking the supremum of the left hand side over all
$t\in[0,T]$ and taking then the limes superior as $\tau \rightarrow
0$, we obtain
\begin{align}
\label{eq:II.42}
  \gamma \limsup_{\tau \rightarrow 0}\sup_{t\in[0,T]}\Vert
  \widetilde{U}_\tau(t)-\underline{U}_\tau(t) \Vert \leq M,
\end{align} 
for any $\gamma>0$, which implies necessarily $\lim_{\tau \rightarrow
  0}\sup_{t\in[0,T]}\Vert \widetilde{U}_\tau(t)-\underline{U}_\tau(t)
\Vert=0$. Since \eqref{eq:II.42} holds for every $t\in [0,T]$, it is
particularly satisfied for $t=t_n$, $n=1,\dots,N$, so that we also get
$\lim_{\tau \rightarrow 0}\sup_{t\in[0,T]}\Vert
\overline{U}_\tau(t)-\underline{U}_\tau(t) \Vert=0$. The latter
convergence in turn implies finally $\lim_{\tau \rightarrow
  0}\sup_{t\in[0,T]}\Vert \widehat{U}_\tau(t)
-\overline{U}_\tau(t)\Vert=0$ which completes the proof.
\end{proof}

\subsection{Limit passage and completion of the proof}
\label{su:Proof}

The next step in constructing a solution to our
\textsc{Cauchy}-Problem relies on compactness arguments in order to
show the existence of a limit function, which obeys the differential
inclusion (1.1) and satisfies the initial datum. For this, it is
natural to make use of the fact that the interpolation functions are
contained in a sublevel set of the energy functional, which by
hypothesis is compact. We elaborate on this in the following result,
which provides also the characterization of the limit function by
\textsc{Young} measures.

\begin{lem} \label{le:LimitPass} Under the same assumptions of Lemma
  \ref{le:Main.Lem}, 
 let $u_0 \in D$ and $(\tau_n)_{n \in \mathbb{N}}$ be a
  vanishing sequence of positive real numbers. Then, there exists a
  subsequence $(\tau_{n_k})_{k \in \mathbb{N}}$, a absolutely
  continuous curve $u\in \AC([0,T];V)$ with $u(0)=u_0$, an integrable
  function $\widetilde{\xi}\in \rmL^1(0,T;V^*)$, a function
  $\mathscr{E}:[0,T]\rightarrow \mathbb{R}$ of bounded variation, an
  essentially bounded function $\mathscr{P}\in \rmL^{\infty}(0,T)$,
  and a time-depended \textsc{Young} measure
  $\mathbold{\mu}=(\mu_t)_{t\in[0,T]}\in \mathscr{Y}(0,T;V\times
  V^*\times \mathbb{R})$, such that
\begin{subequations}
\label{eq:LP.all}
\begin{align}
\label{eq:LP.u}
\overline{U}_{\tau_{n_k}},\underline{U}_{\tau_{n_k}},\widetilde{U}_{\tau_{n_k}},\widehat{U}_{\tau_{n_k}} \rightarrow u \quad \text{in } &\rmL^{\infty}(0,T;V),\\
\label{eq:LP.ud}
\widehat{U}'_{\tau_{n_k}}\rightharpoonup u' \quad \text{in } &\rmL^1(0,T;V),\\
\label{eq:LP.xi}
\widetilde{\xi}_{\tau_{n_k}}\rightharpoonup \widetilde{\xi} \quad \text{in } &\rmL^1(0,T;V^*),\\
\label{eq:LP.B}
B(\teu_{\tau_{n_k}},\underline{U}_{\tau_{n_k}})\rightarrow B(\cdot,u(\cdot))\quad \text{in } &\rmL^{\infty}(0,T;V^*),\\
\label{eq:LP.Ed}
\partial_t \calE_t(\widetilde{U}_{\tau_{n_k}}(t))\rightharpoonup^* \mathscr{P} \quad \text{in } &\rmL^{\infty}(0,T),
\end{align}
\end{subequations} 
and 
\begin{align}
\begin{cases}
\label{eq:LP.E.ptw}
\calE_t(\overline{U}_{\tau_{n_k}}(t))\rightarrow \mathscr{E}(t) \quad &\text{for all }t\in[0,T], \quad \calE_0(u_0)=\mathscr{E}(0),\\
\calE_t(u(t))\leq \mathscr{E}(t)\quad &\text{for all }t\in [0,T]\\
\calE_t(u(t))= \mathscr{E}(t) \quad &\text{for a.a. } t\in (0,T),
\end{cases}
\end{align} 
as $k\rightarrow \infty$. Furthermore, there holds
\begin{subequations}
\label{eq:YM.all}
\begin{align}
\label{eq:YM.ud}
u'(t)&= \int_{V\times V^* \times \mathbb{R}} v \, \dd \mu_t(v,\zeta, p) \quad \text{ for a.a. }t\in[0,T],\\ 
\label{eq:YM.xi}
\widetilde{\xi}(t)&= \int_{V\times V^* \times \mathbb{R}} \zeta \, \dd \mu_t(v,\zeta, p) \quad \text{ for a.a. }t\in[0,T],\\ 
\label{eq:YM.Ed}
\mathscr{P}(t)&= \int_{V\times V^* \times \mathbb{R}} p\, \dd \mu_t(v,\zeta, p)\leq \partial_t \calE_t(u(t))  \quad \text{ for a.a. }t\in[0,T].
\end{align}
\end{subequations} 
and the following energy inequality
\begin{align}
\label{eq:EI}
\begin{split}
&\int_s^t \left( \Psi_{u(r)}(u'(r))+\Psi^*_{u(r)}(B(r,u(r))-\widetilde{\xi}(r))  \right) \dd r +\mathscr{E}(t) \\
&\leq\int_s^t \int_{V\times V^* \times \mathbb{R}} \left( \Psi_{u(r)}(v)+\Psi^*_{u(r)}(B(r,u(r))-\zeta)  \right) \dd \mu_r(v,\zeta, p) \, \dd r +\mathscr{E}(t) \\
&\leq \mathscr{E}(s)+\int_s^t \mathscr{P}(r)\dd r+ \int_s^t \langle B(r,u(r)),u'(r) \rangle \dd r\\
&\leq \mathscr{E}(s)+\int_s^t \partial_r \calE_r(u(r))\dd r+ \int_s^t \langle B(r,u(r)),u'(r) \rangle \dd r
\end{split}
\end{align} for all $0\leq s<t\leq T$.
\end{lem}
\begin{proof} Let the initial datum $u_0 \in D$ and the sequence
$(\tau_n)_{n \in \mathbb{N}}$ of vanishing time steps be given, such
that $\tau_n<\tau^*$ for all $n\in \mathbb{N}$. In order to show the
existence of an absolutely continuous function, we employ the
\textsc{Arzel\`a-Ascoli} theorem on the family of continuous functions
$(\widehat{U}_{\tau_n})_{n\in \mathbb{N}}\subset C([0,T];V)$. First,
we notice that the uniform integrability of
$(\widehat{U}'_{\tau_n})_{n\in \mathbb{N}}$ leads to the
equicontinuity of $(\widehat{U}_{\tau_n})_{n\in \mathbb{N}}$. Second,
the fact that the set $\lbrace
\overline{U}_{\tau_n}(t)\rbrace_{t\in[0,T]}$ belongs to a sublevel set
of the energy functional $E$ for all $n\in \mathbb{N}$, which by
Assumption \ref{eq:cond.E.2} are compact, implies by \textsc{Mazur}'s
lemma that the set $\lbrace
\overline{U}_{\tau_n}(t)\rbrace_{t\in[0,T]}$ also belongs for all
$n\in\mathbb{N}$ to an compact subset of $V$. Therefore by
\textsc{Arzel\`a-Ascoli}, there exists a continuous function $u\in
C([0,T];V)$ such that $\Vert \widehat{U}_{\tau_n}-u
\Vert_{C([0,T];V)}\rightarrow 0$ as $k\rightarrow \infty$ so that in
particular $u(0)=u_0$. Then, the convergences in \eqref{eq:LP.u}
follows from those in \eqref{eq:II.31}.

Further, from the \textsc{Dunford-Pettis} theorem, see
e.g. \textsc{Dunford} and \textsc{Schwartz} \cite[Cor.\,11,
p.\,294]{DunSch59LO1}, which can be applied since both $V$ and $V^*$
are reflexive \textsc{Banach} spaces, we obtain with the uniform
integrability of $(\widehat{U}'_{\tau_n})_{n\in \mathbb{N}}$ and
$(\widetilde{\xi}_{\tau_n})_{n\in \mathbb{N}}$ in $\rmL^1(0,T;V)$ and
$\rmL^1(0,T;V^*)$, respectively, the existence of a subsequence
(labeled as before) and weak limits $v\in \rmL^1(0,T;V)$ and
$\widetilde{\xi} \in \rmL^1(0,T;V^*)$ such that
$\widehat{U}'_{\tau_n}\rightharpoonup v$ weakly in $\rmL^1(0,T;V)$ and
$\widetilde{\xi}_{\tau_n}\rightharpoonup \widetilde{\xi}$ weakly in
$\rmL^1(0,T,V^*)$ as $n\rightarrow \infty$. From a well known
argument, one can identify $v$ as weak derivative of $u$, i.e., $u'=v$
in the weak sense which yields $u\in \rmW^{1,1}(0,T;V)$ and due to
continuity of $u$, $u\in \AC([0,T];V)$.

Now, we shall prove the convergence \eqref{eq:LP.B} of the
perturbation. We first note that the functions $t\mapsto B(t,u(t))$
and $t\mapsto B(\teu_{\tau_{n_k}}(t),\underline{U}_{\tau_{n_k}}(t))$
both belongs to the space $\rmL^{\infty}(0,T;V)$, where the
measurability follows from the continuity of $u$ and $B$, and
Assumptions \ref{eq:B.1} together with \eqref{eq:II.55}, respectively,
whereas the (essential) boundedness is a consequence of Assumptions
\ref{eq:B.2} and \ref{eq:Psi.2} as well as the a priori
estimates. Now, since the interpolation functions are contained in a
sublevel set of the energy functional, uniformly in $\tau>0$ and for
all $t\in(0,T)$, it is also contained in a compact set of $V$,
uniformly in $\tau>0$ and for all $t\in(0,T)$. Therefore, there exists
a compact set $ \mathcal{K}\subset V$ such that by
\textsc{Tychonoff's} theorem the set $[0,T]\times \mathcal{K}$ is
compact with respect to the product topology of $[0,T]\times V$. This,
in turn implies with Assumption \ref{eq:B.1} the uniform continuity of
the map $(t,u)\mapsto B(t,u)$ on $[0,T]\times \mathcal{K}$. Together
with the convergence of
$(\teu_{\tau_{n_k}}(t),\underline{U}_{\tau_{n_k}}(t))) \rightarrow
(t,u(t))$ uniformly in $t\in(0,T)$, we obtain
\begin{align}
\label{eq:II.53}
  \lim_{n\rightarrow \infty}\sup_{t\in(0,T)} \Vert
  B(\teu_{\tau_{n_k}}(t),\underline{U}_{\tau_{n_k}}(t)) -B(t,u(t))\Vert_*
  \quad \quad \text{as } n\rightarrow \infty.
\end{align} 

In order to show the convergence in \eqref{eq:LP.Ed}, we notice that
due to \eqref{eq:II.29} there holds $(\partial_t
\calE_t(\widetilde{U}_{\tau_{n_k}}))_{k\in \mathbb{N}}\subset
\rmL^{\infty}(0,T)$. Since the \textsc{Lebesgue }space
$\rmL^{\infty}(0,T)$ is the dual space of a separable \textsc{Banach}
space $\rmL^1(0,T)$ there exists a limit $\mathscr{P}\in
\rmL^{\infty}(0,T)$ such that (up to a subsequence) $\partial_t
\calE_t(\widetilde{U}_{\tau_{n_k}}) \rightharpoonup^* \mathscr{P}$
weakly$^*$ in $\rmL^{\infty}(0,T)$ as $k\rightarrow \infty$.  \bigskip

Now, we shall prove \eqref{eq:LP.E.ptw}. For this, we define
\begin{align*}
  \eta_\tau(t):= \calE_{\teo_\tau(t)}
  (\overline{U}_\tau(t))-\int_0^{\teo_\tau(t)} \partial_r
  \calE_r(\widetilde{U}_\tau(r))\dd r-\int_0^{\teo_\tau(t)} \langle
  B(\teu_\tau(r),\underline{U}_\tau(r)), \widehat{U}'_\tau (r) \rangle
  \dd r
\end{align*} 
for $t\in[0,T]$ and we deduce from the discrete upper energy estimate
\eqref{eq: DUEE} that the map $t\mapsto \eta_\tau(t): [0,T]\rightarrow
\mathbb{R}$ is non-increasing. Then, by \textsc{Helly}s theorem there
exists a non-increasing function $\eta:[0,T]\rightarrow \mathbb{R}$
and a subsequence (labeled as before) such that $\eta_{\tau_{n_k}}(t)
\rightarrow \eta(t)$ as $k \rightarrow \infty$ for all
$t\in[0,T]$. Moreover, we define
\begin{align*}
  \psi_\tau(t):= \int_0^{\teo_\tau(t)} \langle
  B(\teu_\tau(r),\underline{U}_\tau(r)), \widehat{U}'_\tau (r) \rangle
  \dd r \quad \text{ for } t\in [0,T].
\end{align*} 
Since we have strong convergence of the perturbation
$B(\teu_\tau,\underline{U}_{\tau_{n_k}})$ in $\rmL^{\infty}(0,T;V^*)$
and weak convergence of the derivative $\widehat{U}'_{\tau_{n_k}} $ in
$\rmL^1(0,T;V) $ as $k\rightarrow \infty$, there holds
\begin{align}
\label{eq:II.54}
\psi_{\tau_{n_k}}(t)\rightarrow \psi(t):=\int_0^t \langle B(r,u(r)),
u'(r) \rangle \dd r \quad \text{ as } k\rightarrow \infty
\end{align} 
for all $t\in[0,T]$. Considering convergence \eqref{eq:LP.Ed}, we obtain 
\begin{align*}
  \calE_{\teo_{\tau_{n_k}}(t)}(\overline{U}_{\tau_{n_k}}(t))\rightarrow
  \mathscr{E}(t):= \eta(t)+\int_0^t \mathscr{P}(r)\dd r +\psi(t) \quad
  \text{for all } t\in[0,T]
\end{align*} 
as $k\rightarrow \infty$. Since the function $\eta$ is
monotone and both the function $\psi$ and the map $t\mapsto \int_0^t
\mathscr{P}(r)\dd r$ are absolutely continuous, it follows that the
function $\mathscr{E}$ is of bounded variation. In order to conclude
the convergence in \eqref{eq:LP.E.ptw}, we notice that
\begin{align*}
  \vert \calE_{\teo_{\tau_{n_k}}(t)}(\overline{U}_{\tau_{n_k}}(t)) -
  \calE_t(\overline{U}_{\tau_{n_k}}(t))\vert \rightarrow 0 \quad
  \text{as } k\rightarrow \infty
\end{align*} 
which follows from \eqref{eq:II.5}, \eqref{eq:II.29} and the fact that
$\teo_{\tau_{n_k}}(t)\rightarrow t$ as $k \rightarrow \infty$ for all
$t\in[0,T]$. Further, by the lower semicontinuity of the energy
functional, we obtain due to the convergence \eqref{eq:II.31}
\begin{align}
\label{eq:II.55}
\calE_t(u(t))\leq \liminf
\calE_t(\overline{U}_{\tau_{n_k}}(t))=\mathscr{E}(t)\leq M \quad
\text{for all } t\in[0,T],
\end{align} 
where the last inequality follows from \eqref{eq:II.29}. The last
assertion in \eqref{eq:LP.E.ptw} follows from Assumption
\ref{eq:cond.E.5}.
\medskip

We continue by showing \eqref{eq:YM.all}. For this
purpose, we define the (reflexive) \textsc{Banach} space
$\mathcal{V}:=V\times V^*\times \mathbb{R}$ endowed with the product
topology space and employ the fundamental theorem of weak topologies
(Theorem \ref{th:A2}) applied to the sequence
$w_k:=(\widehat{U}'_{\tau_{n_k}},
\widetilde{\xi}_{\tau_{n_k}}, \partial_t \calE_t
(\widetilde{U}_{\tau_{n_k}}) )_{k\in \mathbb{N}}$ which belongs to
$\rmL^1(0,T;\mathcal{V})$ by the a priori estimates, and is uniformly
integrable in $\rmL^1(0,T;\mathcal{V})$ since every component is in
the respective space. Thus, there exists a \textsc{Young}-measure
$\mathbold{\mu}=(\mu_t)_{t\in[0,T]}\in \mathscr{Y}(0,T;V\times
V^*\times \mathbb{R})$ such that $\mu_t$ is for almost everywhere
$t\in(0,T)$ concentrated on the set
\begin{align*}
  \mathrm{Li}(t):=\bigcap_{p=1}^\infty \mathrm{clos_{weak}}
  \big(\lbrace w_k(t) : k\geq p \rbrace \big)
\end{align*}
of all limit points of $w_k(t)$ with respect to the weak-weak-strong
topology of $V\times V^*\times \mathbb{R}$, i.e.\
$\mathrm{sppt}(\mu_t) \subset \mathrm{Li}(t) $. Since the weak limits
in \eqref{eq:LP.ud}, \eqref{eq:LP.xi} and \eqref{eq:LP.Ed} are unique,
the identities in \eqref{eq:YM.ud} and \eqref{eq:YM.xi} are direct
consequences of the fundamental theorem of weak topologies, whereas
the inequality in \eqref{eq:YM.Ed} is true due to the fact that for
almost every $t\in (0,T)$, there holds
\begin{align}
\label{eq:II.56}
\zeta \in \partial \calE_t(u(t)) \quad \text{and} \quad
p\leq \partial_t \calE_t(u(t)) \quad \text{for all } (v,\zeta,p)\in
\mathrm{Li}(t).
\end{align} 
Property \eqref{eq:II.56} in turn follows from Assumption
\ref{eq:cond.E.5} with the convergences in
\eqref{eq:LP.u}{eq:LP.all} and \eqref{eq:LP.E.ptw} as well as the inclusion
\eqref{eq:II.26}: Let $\mathcal{N}\subset (0,T)$ a negligible set such
that for all $t\in (0,T)\backslash \mathcal{N}$ the set
$\mathrm{Li}(t)$ is non-empty. Now let $t\in t\in (0,T)\backslash
\mathcal{N}$ and $(v,\zeta,p)\in \mathrm{Li}(t)$, then there exists a
subsequence $(k_l)_{l\in \mathbb{N}}$ such that
$\widehat{U}'_{\tau_{n_{k_l}}}(t) \rightharpoonup v,\,
\widetilde{\xi}_{\tau_{n_{k_l}}}(t)\rightharpoonup^* \zeta$ and
$\partial_t \calE_t(\widetilde{U}_{\tau_{n_{k_l}}}(t))\rightarrow p$
as $l\rightarrow \infty$, where the latter convergence follows from
the fact that in finite dimensional spaces the weak topology coincides
with the strong topology. In view of convergence \eqref{eq:LP.u} and
the inclusion \eqref{eq:II.26}, \eqref{eq:II.56} follows by Assumption
\ref{eq:cond.E.5}. Integrating the inequality in \eqref{eq:II.56} with
respect to the \textsc{Borel} probability measure yields
\eqref{eq:YM.Ed}. In order to show the energy inequality
\eqref{eq:EI}, we notice first of all that from \textsc{Jensen}'s
inequality, we obtain for almost every $t\in(0,T)$
\begin{align}
\label{eq:JI.psi}
\Psi_{u(t)}(u'(t))&\leq \int_{V\times V^*\times \mathbb{R}} \Psi_{u(t)}(v)\dd \mu_t(v,\zeta,p),\\
\label{eq:JI.psi*}
\Psi^*_{u(t)}(B(t,u(t))-\widetilde{\xi}(t)))&\leq \int_{V\times V^*\times \mathbb{R}} \Psi^*_{u(t)}(B(t,u(t))-\zeta)\dd \mu_t(v,\zeta,p).
\end{align} 
This can also be obtained by integrating the inequalities
\begin{align*}
  \Psi_{u(t)}(u'(t))&\leq \Psi_{u(t)}(v)+\langle w^*,u'(t)-v\rangle
  \quad \text{for all }v\in V \\
  \Psi^*_{u(t)}(B(t,u(t))-\widetilde{\xi}(t))&\leq
  \Psi^*_{u(t)}(B(t,u(t))-\zeta)+\langle \zeta-\widetilde{\xi}(t)
  ,w\rangle \quad \text{for all }\zeta \in V^*
\end{align*} 
using the identities in \eqref{eq:YM.all} as well as
the fact that $w^*\in \partial \Psi_{u(t)}(u'(t))\neq \emptyset$ and
$w\in \partial \Psi^*_{u(t)}(B(t,u(t))-\widetilde{\xi}(t))\neq
\emptyset$, see Remark \ref{re:Assump.Psi} i).
\bigskip

Defining $ \mathcal{H}_k:[0,T]\times \mathcal{V}\rightarrow \mathbb{R}$ by
\begin{align*}
\mathcal{H}_k(r,w):=\chi_{[\teo_{\tau_{n_k}}(s),\teo_{\tau_{n_k}}(t)]} \Psi_{\underline{U}_{\tau_{n_k}}(r)}(v), \quad (r,v,\zeta,p)\in [0,T]\times \mathcal{V},
\end{align*} together with \eqref{eq:II.29} and \eqref{eq:LP.u}, the \textsc{Mosco} continuity \ref{eq:Psi.3} leads to 
\begin{align}
\label{eq:II.59}
\mathcal{H}(r,w):=\chi_{[s,t]}\Psi_{u(r)}(v)\leq \liminf_{k\rightarrow \infty} \mathcal{H}_k(r,w_n),
\end{align} for all $(r,w)=(r,v,\zeta,p)\in [0,T]\times \mathcal{V}$ and all weak convergent sequences $w_k\rightharpoonup w\in \mathcal{V}$, where $s,t \in [0,T]$ with $s\leq t$ are chosen to be fixed. As the space \textsc{Banach} space $\mathcal{V}$ is reflexive, the map
\begin{align*}
(v,\zeta,p) \mapsto (\Vert v \Vert+\Vert \zeta \Vert_{*}+\vert p \vert)
\end{align*} has compact sublevel sets with respect to the weak
topology of $\mathcal{V}$. Together with the boundedness of the
afore-defined sequence $(w_k)_{k\in \mathbb{N}}$, which follows from
\eqref{eq:LP.all}, we obtain the weak-tightness of $(w_k)_{k\in
  \mathbb{N}}$. Therefore, for a subsequence of $(n_k)_{k\in
  \mathbb{N}}$ (not relabeled), Theorem \ref{th:A1} provides the
inequality
\begin{align*}
\int_0^T \int_\mathcal{V} \mathcal{H}(r,w) \dd \mu_r(w)\dd r \leq \liminf_{k\rightarrow \infty} \int_0^T  \mathcal{H}_k(r,w_k) \dd r, 
\end{align*} i.e.,
\begin{align}
\label{eq:LI.psi}
\int_s^t \int_{\mathcal{V}}\Psi_{u(r)}(v)\dd \mu(v,\zeta,p) \dd r\leq \liminf_{k\rightarrow \infty} \int_{\teo_{\tau_{n_k}}(s)}^{\teo_{\tau_{n_k}}(t)} \Psi_{\underline{U}_{\tau_{n_k}}(r)}(\widehat{U}'_{\tau_{n_k}}(r))\dd r<+\infty,
\end{align} where the boundedness follows from the a priori estimate \eqref{eq:II.30}. Taking into account Remark \ref{re:Assump.Psi}  $iii)$, then Theorem \ref{th:A1} applied to the function
\begin{align*}
\mathcal{H}^*_k(r,w):=\chi_{[\teo_{\tau_{n_k}}(s),\teo_{\tau_{n_k}}(t)]} \Psi^*_{\underline{U}_{\tau_{n_k}}(r)}(B(\teu_{\tau_{n_k}}(r),\underline{U}_{\tau_{n_k}}(r))-\zeta), \quad (r,v,\zeta,p)\in [0,T]\times \mathcal{V},
\end{align*} yields 
\begin{align}
\label{eq:LI.psi*}
\begin{split}
&\int_s^t \int_\mathcal{V}\Psi^*_{u(t)}(B(r,u(r))-\zeta)\dd \mu(v,\zeta,p) \dd r\\
&\leq \liminf_{k\rightarrow \infty} \int_{\teo_{\tau_{n_k}}(s)}^{\teo_{\tau_{n_k}}(t)} \Psi^*_{\underline{U}_{\tau_{n_k}}(r)}(B(\teu_{\tau_{n_k}}(r),\underline{U}_{\tau_{n_k}}(r))-\widetilde{\xi}_{\tau_{n_k}}(r))\dd r<+\infty,
\end{split}
\end{align} where again the boundedness follows from \eqref{eq:II.30}.
Integrating \eqref{eq:JI.psi} and \eqref{eq:JI.psi*} with respect to
$t$ yields the first inequality in \eqref{eq:EI}. The second and third
inequality follow by passing to the limit in the discrete upper energy
estimate \eqref{eq: DUEE} as $k\rightarrow \infty$ and considering
\eqref{eq:LP.Ed}, \eqref{eq:LP.E.ptw}, \eqref{eq:YM.Ed},
\eqref{eq:II.54}, \eqref{eq:II.56} as well as \eqref{eq:LI.psi} and
\eqref{eq:LI.psi*}. This proves Lemma \ref{le:LimitPass}.
\end{proof}

 We are now ready to complete the proof of our main existence result in
Theorem \ref{th:MainExist}. 

\begin{proof}[Proof of Theorem \ref{th:MainExist}]
  In order to show that the absolutely continuous curve $ $ $u\in
  \AC([0,T];V)$ obtained from Lemma \ref{le:LimitPass} is a solution to
  the differential inclusion \eqref{eq:I.1}, we make use of the chain
  rule for \textsc{Young} measures in Lemma \ref{le:A3} which is
  justified by \eqref{eq:LP.Ed}, \eqref{eq:YM.ud}, \eqref{eq:II.56},
  \eqref{eq:LI.psi} and \eqref{eq:LI.psi*}, where
  $\mathbold{\mu}=(\mu_t)_{t\in[0,T]}\in \mathscr{Y}(0,T;V\times
  V^*\times \mathbb{R})$ is to be chosen as in Lemma
  \ref{le:LimitPass}. Hence by the chain rule condition, the map $t
  \mapsto \calE_t(u(t))$ is absolutely continuous on $(0,T)$ and there
  holds
\begin{align*}
  \frac{\rmd}{\rmd t}\calE_t(u(t))\geq \int_{V\times V^*\times
    \mathbb{R}}\langle \zeta,u'(t) \rangle \dd
  \mu_t(v,\zeta,p)+\partial_t \calE_t(u(t))\quad \text{for
    a.a. }t\in(0,T).
\end{align*} 
Thus, together with \eqref{eq:LP.E.ptw}, \eqref{eq:YM.Ed} and
\eqref{eq:EI}, we obtain with $s=0$
\begin{align}
\label{eq:II.62}
\begin{split}
&\int_0^t \int_{V\times V^* \times \mathbb{R}} \left( \Psi_{u(r)}(u'(r))+\Psi^*_{u(r)}(B(r,u(r))-\zeta)  \right) \dd \mu_r(v,\zeta, p) \, \dd r +\calE_t(u(t)) \\
&\leq \calE_0(u_0)+\int_0^t  \partial_r \calE_r(u(r)) \dd r+ \int_0^t \langle B(r,u(r)),u'(r) \rangle \dd r\\
&\leq \calE_t(u(t))-\int_0^t \int_{V\times V^* \times \mathbb{R}} \langle  \zeta,u'(r) \rangle \dd \mu_r(v,\zeta,p)\dd r+\int_0^t \langle B(r,u(r)),u'(r) \rangle \dd r\\
&= \calE_t(u(t))+\int_0^t \int_{V\times V^* \times \mathbb{R}} \langle B(r,u(r))-\zeta,u'(r) \rangle \dd \mu_r(v,\zeta,p) \dd r \quad \text{for all }t\in[0,T].
\end{split}
\end{align} Therefore, there holds
\begin{align}
\label{eq:II.63}
\begin{split}
\int_0^t \int_{V\times V^* \times \mathbb{R}} &( \Psi_{u(r)}(u'(r)) +\Psi^*_{u(r)}(B(r,u(r))-\zeta) \\  &-\langle B(r,u(r))-\zeta,u'(r) \rangle ) \dd \mu_r(v,\zeta, p) \dd r \leq 0 \quad \text{for all }t\in[0,T].
\end{split}
\end{align} Then, from the \textsc{Fenchel-Young} inequality we deduce the non-negativity of the integrand in \eqref{eq:II.63} and infer therefore
\begin{align}
\label{eq:II.64}
& \int_{V\times V^* \times \mathbb{R}} \left( \Psi_{u(t)}(u'(t))+\Psi^*_{u(t)}(B(t,u(t))-\zeta) -\langle B(t,u(t))-\zeta,u'(t) \rangle \right) \dd \mu_t(v,\zeta, p)\notag \\
&=0 \quad \text{for a.a. }t\in(0,T).
\end{align} It follows that all inequalities in \eqref{eq:II.62} become equalities for all $t\in[0,T]$, so that we obtain the equation
\begin{align}
\label{eq:II.65}
\begin{split}
&\int_s^t \int_{V\times V^* \times \mathbb{R}} \left( \Psi_{u(r)}(u'(r))+\Psi^*_{u(r)}(B(r,u(r))-\zeta)  \right) \dd \mu_r(v,\zeta, p) \, \dd r +\calE_t(u(t)) \\
&= \calE_s(u(s))+\int_s^t  \partial_r \calE_r(u(r)) \dd r+ \int_s^t \langle B(r,u(r)),u'(r) \rangle  \dd r 
\end{split} 
\end{align} for all $0\leq s,t\leq T$. Defining the marginal $\mathbold{\nu}=(\nu_t)_{t\in[0,T]}:=\pi^{2,3}_\#\mathbold{\mu}$ of $\mathbold{\mu}$ by $\nu_t(B):=\mu_t((\pi^{2,3})^{-1}(B))$ for all $B\in \mathscr{B}(V^*\times \mathbb{R})$, where $\pi^{2,3}:V\times V^*\times \mathbb{R}\rightarrow V^*\times \mathbb{R}$ denotes the canonical projection and $\mathscr{B}(V^*\times \mathbb{R})$ the \textsc{Borel} $\sigma$-algebra of $V^*\times \mathbb{R}$. Setting 
\begin{align}
\label{eq:II.66}
\begin{split}
\mathcal{S}(t,u(t),u'(t)):=\lbrace &(\zeta,p)\in V^*\times \mathbb{R} \mid \zeta \in \partial
 \calE_t(u(t)) \cap (B(t,u(t))-\partial\Psi_{u(t)}(u'(t))\\
 & \text{and } p\leq \partial_t\calE_t(u(t)) \rbrace
\end{split}
\end{align} we notice that by \eqref{eq:II.56} and \eqref{eq:II.64} it follows that $\nu_t(\mathcal{S}(t,u(t),u'(t)))=1$ for a.a. $t\in(0,T)$ and assumption \eqref{eq:A7} is fulfilled. Therefore, by Lemma \ref{le:A4} there exists a measurable selections $\xi:[0,T]\rightarrow V^*$ and $p:[0,T]\rightarrow \mathbb{R}$ with
\begin{align}
\label{eq:II.67}
\int_0^T \Psi^*_{u(t)}(B(t,u(t))-\xi(t))\dd t<+\infty,
\end{align} such that $(\xi(t),p(t))\in \mathcal{S}(t,u(t),u'(t))$ and there holds
\begin{align}
\label{eq:II.68}
\Psi^*_{u(t)}(B(t,u(t))-\xi(t))-p(t)=\min_{(\zeta,p)\in \mathcal{S}(t,u(t),u'(t))} \Psi^*_{u(t)}(B(t,u(t))-\zeta)-p
\end{align} Since \eqref{eq:II.67} holds and $B(\cdot,u(\cdot))\in \rmL^{\infty}(0,T;V^*)$, we deduce from Assumption from the superlinearity of $\Psi_u^*$ that $ $  $\xi\in \rmL^1(0,T;V^*)$, so that the pair $(u,\xi)$ solves the differential inclusion \eqref{eq:I.1} and $u$ satisfies the initial condition $u(0)=u_0$, where the former follows from \eqref{eq:II.68} and the latter by Lemma \ref{le:LimitPass}.\\  Furthermore, taking into account property \eqref{eq:II.56} and equation \eqref{eq:II.64}, then Lemma \ref{le:Leg.Fen} yields $\nu_t(\mathcal{S}(t,u(t),u'(t))=1$ for almost every $t\in(0,T)$. Thus from equality \eqref{eq:II.68} and the definition of $\mathcal{S}(\cdot,u(\cdot).u'(\cdot))$, there holds
\begin{align*}
&\int_s^t \Psi^*_{u(r)}(B(r,u(r)-\xi(r))\dd r -\int_s^t p(r)\dd r\\
& \leq \int_s^t \int_{V\times V^*\times \mathbb{R}} \Psi^*_{u(r)}(B(r,u(r))-\zeta)\dd \mu_r(v,\zeta,p)\dd r - \int_s^t p(r) \dd r
\end{align*} Now, by comparison with equation \eqref{eq:II.65}, we infer
\begin{align*}
\begin{split}
&\int_s^t \left( \Psi_{u(r)}(u'(r))+\Psi^*_{u(r)}(B(r,u(r))-\xi(r))  \right)\dd r +\calE_t(u(t)) \\
&\leq  \calE_s(u(s))+\int_s^t  \partial_r \calE_r(u(r)) \dd r+ \int_s^t \langle B(r,u(r)),u'(r) \rangle  \dd r 
\end{split}
\end{align*} for all $0\leq s\leq t\leq T$. On the other hand, applying the chain rule condition \ref{eq:cond.E.4} to the pair $(u,\xi)$ yields
\begin{align*}
\frac{\rmd}{\rmd t} \calE_t(u(t))\geq \langle \xi(t),u'(t)\rangle + \partial_t \calE_t(u(t)) \quad \text{ for a.e. }t\in (0,T). 
\end{align*} Together with the identity
\begin{align*}
\Psi_{u(r)}(u'(r))+\Psi^*_{u(r)}(B(r,u(r))-\xi(r))  
= \langle B(r,u(r))-\xi(r),u'(r) \rangle  \quad \text{ a.e. in} \in (0,T),
\end{align*} which again follows from Lemma \ref{le:Leg.Fen} and the definition of $\mathcal{S}(\cdot,u(\cdot).u'(\cdot))$, we conclude the energy-dissipation balance \eqref{eq:EDB}.
\end{proof}
\begin{rem}
It is not difficult to prove that for every sequence $(\tau_n)_{n\in\mathbb{N}}$ there exists a subsequence (denoting as before) such that the following convergences holds:
\begin{align*}
  \calE_t(\overline{U}_{\tau_n}(t))&\rightarrow \calE_t(u(t))  \quad \text{for all }t\in[0,T],\\
  \int_s^t
  \Psi_{\underline{U}_{\tau_n}(r)}(\widehat{U}_{\tau_n}'(r))\dd r &
  \rightarrow  \int_s^t\Psi_{u(r)}(u'(r)) \dd r \quad \text{and }\\
  \int_s^t
  \Psi^*_{\underline{U}_{\tau_n}(r)}
  (B(\teu_{\tau_n}(r),\underline{U}_{\tau_n}(r)) -\widetilde{\xi}_{\tau_n}(r))
  \dd r &\rightarrow \int_s^t\Psi^*_{u(r)}(B(r,u(r))-\xi(r))\dd r
\end{align*} for all $0\leq s\leq t\leq T$ as $n\rightarrow \infty$. Furthermore, if we additionally assume that the dissipation potential $\Psi_u$ and its conjugate $\Psi^*_u$ are strictly convex for all $u\in V$, then there holds $\pi^1_\# \mathbold{\mu}=\delta_{u'(t)}$ and $\pi^2_\# \mathbold{\mu}=\delta_{\xi(t)}$, respectively, and there holds
\begin{align*}
\widehat{U}_{\tau_n}'(t) \rightharpoonup u'(t) \quad \text{and} \quad \widetilde{\xi}_{\tau_n}(t)\rightharpoonup \xi(t) \quad \text{for a.a. }t\in(0,T).
\end{align*} as well as $\widetilde{\xi}_{\tau_n}\rightharpoonup \xi$ in $\rmL^1(0,T;V^*)$ as $n\rightarrow \infty$.
\end{rem}

\section{A result for evolutionary $\Gamma$-convergence}
\label{se:EGC}

In this section we consider a family of perturbed gradients systems
$\PG^\eps:=(V,\calE^\eps, \Psi^\eps,B^\eps)$, where $\eps\in [0,1]$ is a small
parameter. Here the case $\eps=0$ is the supposed limit equation, also
called effective equation. The major question what type of convergence
of $\calE^\eps$, $\Psi^\eps$, and $B^\eps$ is sufficient to conclude
that solutions $u_\eps:[0,T] \to V$ for $\PG^\eps$ with $\eps>0$ 
have subsequences $\eps_k\to 0$ that convergence
pointwise in $t\in [0,T]$ to a limit function $u_0:[0,T]\to V$ and
that $u_0$ is indeed a solution for $\PG^0$. 

The theory developed here follows \cite[Thm.\,4.8]{MiRoSa13NADN},
where the case of pure gradient systems (i.e.\ $B_\eps \equiv 0$) was
considered.

\subsection{Assumptions and results}
\label{su:EGC.Ass.Res}

Our assumptions follow closely the assumption for the existence theory
in Section \ref{su:AssumpExistRes}, where we need uniformity with
respect to $\eps\in [0,1]$.  For definiteness we now list the precise
assumptions on $\PG^\eps$. For describing energy functionals
$\calE^\eps$ w define the auxiliary 
\begin{align*}
& G^\eps(u)= \sup\bigset{ \calE^\eps_t(u)}{ t\in [0,T] }\\
&\ulG(u):=\inf\bigset{ \calE_t^\eps(u) }{ t\in[0,T],\ \eps \in
  [0,1] }.
\end{align*}
Without loss of generality we may assume that $\ulG$ is bounded
from below by a positive constant $\gamma>0$.
{\renewcommand{\theequation}{\thesection.E$^\eps$}%
\begin{subequations}\label{eq:cond.E.eps}%
 \begin{align}
  \nonumber
  &\text{\textbf{Constant domains.}} \quad \forall \, t\in[0,T] \
    \forall\, \eps \in [0,1]:\\
  \nonumber
  & \quad \calE^\eps_t:V \rightarrow (0,\infty] \text{ is proper and lower
   semicontinuous with}\\
  &\quad \text{time-independent domain } D^\eps:= 
   \mathrm{dom}(\calE^\eps_t)\subset V  \text{ for all } t\in [0,T].
  \label{eq:cond.E.eps.1} 
 \\
  \nonumber 
  &\text{\textbf{Equi-compactness of sublevels.}} \\
  &\quad\text{The sublevels of } \ulG \text{ have compact closure in }  V. 
  \label{eq:cond.E.eps.2} 
  \\
   \nonumber
   &\text{\textbf{Uniform energetic control of power.}} \\ 
   \nonumber
   & \quad \forall\, \eps \in [0,1] \ \forall \, u\in D^\eps: \quad 
   t\mapsto \calE^\eps_t(u) \text{ is differentiable on } (0,T)
   \text{and } \\
  \label{eq:cond.E.eps.3} 
  & \quad \exists\, C_T>0\ \forall\, \eps \in [0,1]\ \forall\, t\in (0,T)\
    \forall\, u \in D^\eps:\quad   
    \vert \partial_t \calE^\eps_t(u)\vert \leq C_T \calE^\eps_t(u) .
  \\[0.4em]
  \label{eq:cond.E.eps.4} 
  &\text{\textbf{Chain rule.}} \quad \forall\, \eps \in [0,1]:
    \quad\text{the chain rule of \ref{eq:cond.E.4} holds for }
    (V,\calE^\eps,\Psi^\eps). \hspace*{2em} 
  \\[0.4em] 
  &\label{eq:cond.E.epsGamma}
   \text{\textbf{Liminf estimate.}} \ (\eps_k,u_k) \to (0,u) \text{
     implies } \calE^0_t(u)\leq \liminf_{k\to \infty} \calE^{\eps_k}_t(u_k).
  \\[0.4em] 
  &\nonumber
    \text{\textbf{Strong-weak closedness in the limit }} \eps
     \to 0. \quad \text{For all } t\in[0,T] \text{ and}
  \\
  &\nonumber
   \quad \text{all sequences } (\eps_n,u_n,\xi_n)_{n\in \mathbb{N}} 
     \subset [0,1]\ti V\ti V^* \text{ with }
     \xi_n\in \partial \calE^{\eps_n}_t(u_n) \text{ and } 
  \\
  &\nonumber
    \qquad\eps_n\to 0,\ \   u_n \rightarrow u\in V, \ \  \xi_n 
       \rightharpoonup \xi\in V^*,
     \ \  \calE^{\eps_n}_t(u_n)\rightarrow \calE_0,   
     \ \  \partial_t \calE^{\eps_n}_t(u_n)\rightarrow \calP
     \hspace*{-4em}\mbox{ }
  \\  
  &\nonumber 
   \quad \text{for $n\rightarrow\infty$, we have the relations }
  \\
  & \qquad \xi \in \partial \calE^0_t(u), \quad  \calE^0_t(u)=\calE_0, \quad 
     \text{and} \quad  \partial_t \calE^0_t(u)\geq \calP. 
   \label{eq:cond.E.eps.5} 
 \end{align}
\end{subequations}
\addtocounter{equation}{-1}}
%
%
 As in the existence
theory we use a control of the time-derivative, see
\eqref{eq:cond.E.eps.3}, which gives $\calE_t^\eps(u)\geq
\ee^{-C_T|t{-}s|} \calE^\eps_s(u)$. Thus, for all $\eps\in [0,1]$ and
$t\in [0,T]$ we have the relations 
\[
\ulG(u) \leq G^\eps(u) \leq \ee^{C_T T} \calE^\eps_t(u) \leq  \ee^{C_T
  T} G^\eps(u).
\]
Note that we cannot use a uniform upper bound $G^\eps(u) \leq \ol
G(u)$ as this would exclude many useful results on
$\Gamma$-convergence. 

In the present form of condition \eqref{eq:cond.E.eps.5} we do not ask for the
strong-weak closedness for $\calE^\eps_t$ with a given positive
$\eps$. However, in our main result we simply assume the existence of
solutions $u_\eps:[0,T]\to V$ for $\PG^\eps$. If we want to show this
with the theory of Section \ref{se:ExistResult}, then one
has to impose \ref{eq:cond.E.5} for all $\eps>0$ as well (which is the same as
allowing constant sequences $\eps_n=\eps$ in \eqref{eq:cond.E.eps.5}.

The closedness condition \eqref{eq:cond.E.eps.5} looks rather strong,
however in Remark \ref{re:SWClosGamma}, see after the statement of the main
convergence result, we will 
show that convexity of $\calE^\eps_t(\cdot)$ and strong
$\Gamma$-convergence to $\calE^0_t$ already imply the
desired closedness.

The conditions of the dissipation potentials $\Psi^\eps_u:V\to
[0,\infty)$ are the following. 
{\renewcommand{\theequation}{\thesection.$\Uppsi^\eps$}%
\begin{subequations}
 \label{eq:Psi.eps}
 \begin{align}
  \nonumber
   &\textbf{Dissipation potentials.} \quad \forall\, \eps\in [0,1]\
   \forall \, u\in V: \\ 
  \label{eq:Psi.eps.1}
   &\quad \Psi^\eps_u:V\rightarrow [0,\infty) \text{ is lower
     semicontinuous and convex with } \Psi^\eps_u(0)=0.
   \\[0.3em]
  \nonumber
   &\textbf{Superlinearity.} \quad \forall\ R>0 \
     \exists\, g_R:[0,\infty)\to [0,\infty) \text{ superlinear}: \\
   \nonumber
   &\quad \forall\, \eps\in [0,1]\ \forall\, (v,\xi)\in V\ti V^*\ 
     \forall\, u\in V \text{ with } \ulG(u)<R:\\
   & \label{eq:Psi.eps.2} 
     \quad\qquad \Psi^\eps_u(v)\geq g_R(\|v\|)
     \text{ and } \Psi^{\eps,*}_u(\xi) \geq g_R(\|\xi\|). 
 \\
  \nonumber
   & \text{\textbf{Mosco convergence.} \ For all $R>0$ and sequences }
     (\eps_n,u_n)_{n\in \N} \subset [0,1]\ti V\\
  \label{eq:Psi.eps.3} 
   &\quad  \text{ with } \ulG(u_n)\leq R \text{ and }
   (\eps_n,v_n)\to (0,v): \quad 
   \Psi^{\eps_n}_{u_n} \Mto \Psi^0_u .  
\end{align}
\end{subequations}
\addtocounter{equation}{-1}}
%
Again we have formulated the \textsc{Mosco} convergence of the dissipation
potentials only with the limit $\eps_n\to 0$, which is sufficient for
the limit passage when solutions $u_\eps:[0,T]\to V$ are given. To
show the existence of solutions we need \ref{eq:Psi.3} for all $\eps\in
(0,1]$ as well. 
    
Finally, we impose the conditions of the non-variational perturbation
$B^\eps$, namely 
{\renewcommand{\theequation}{\thesection.B$^\eps$}%
\begin{subequations}
 \label{eq:Beps}
  \begin{align} 
  &\label{eq:Beps.1}
    \text{\textbf{Continuity.}  \ The map }\left\{\ba{ccc} [0,1]\ti
        [0,T]\ti V & \to & V^*, \\ \
    (\eps,t,u) &\mapsto& 
    B^\eps (t,u), \ea \right. \text{ is continuous.}
  \\
  & \nonumber
   \text{\textbf{Control of $B^\eps$ by the energy.}} \quad \exists\,
   C_B>0\ \forall\, (\eps, t)\in [0,1]\ti [0,T] \\
    &\quad  \forall\, u\in D^\eps : \quad 
     \Psi^{\eps,*}_u\big(B^\eps(t,u)\big)\leq C_B \calE^\eps_t(u) .
  \label{eq:Beps.2}
\end{align}
\end{subequations} 
\addtocounter{equation}{-1}}

We are now ready to formulate our result of evolutionary
$\Gamma$-converge. In \cite{Miel16EGCG} the convergence
we will established is called ``pE-convergence'' as we have to impose
the well-``p''reparedness of the initial conditions $u^0_\eps$, viz.\ 
\begin{equation}
  \label{eq:EGC.well}
  u^0_\eps \to u^0 \text{ in } V \ \text{ and } \
  \calE^\eps_0(u^0_\eps) \to \calE^0_0(u^0)<\infty \quad \text{ for } \eps
  \to 0.
\end{equation}
Moreover, in the sense of \cite{LMPR17MOGG,DoFrMi17?EGCW} we even have
the much stronger notion of EDP convergence, which means convergence
in the sense of the energy-dissipation balance. Indeed, as for the
existence result in Section \ref{se:ExistResult} we will also strongly
rely on the energy-dissipation principle and perform the limit
$\eps\to 0$ in the energy-dissipation balance \eqref{eq:EDB}. Our proof
will be an adaptation of \cite[Thm.\,4.8]{MiRoSa13NADN}. 

\begin{thm}[Evolutionary $\Gamma$-convergence] \label{th:EGC.main}
Assume that the family $\PG^\eps=(V,\calE^\eps,\Psi^\eps,B^\eps)$,
$\eps\in [0,1]$ satisfy the assumptions \eqref{eq:cond.E.eps},
\eqref{eq:Psi.eps}, and \eqref{eq:Beps}.  Moreover, assume that for
$\eps>0$ we have solutions $u_\eps:[0,T]\to V$ of $\PG^\eps$ such that
the initial conditions $u_\eps(0)=u^0_\eps$ satisfy
\eqref{eq:EGC.well}.
Then, there exists a subsequence $\eps_k\to 0$ and a solution
$u:[0,T]\to V$ of the limit system $\PG^0$ with $u(0)=u^0$ such that
the following convergences hold:
\begin{subequations}
 \label{eq:EGC.cvg}
 \begin{align}
  & \label{eq:EGC.cvg.a}
    u_{\eps_k}(t)\ \to \ u(t)\  \text{ in } \rmC^0([0,T];V); \\
  & \label{eq:EGC.cvg.b} 
   \forall\, t\in [0,T]: \quad \calE^{\eps_k}_t(u_{\eps_k}(t)) \ \to \
    \calE^0_t(u(t)); \\
  & \label{eq:EGC.cvg.c} 
    u'_{\eps_k}\  \weak\  u' \ \text{ in } \rmL^1(0,T;V); \\
  & \label{eq:EGC.cvg.d} 
    \forall\, r<s:\ \int_r^s \Psi^{\eps_k}_{u_{\eps_k}(t)} (u'_{\eps_k}(t)) \dd t \
    \to \int_r^s \Psi^0_{u(t)} (u'(t)) \dd t; \\
  & \label{eq:EGC.cvg.e} \forall\, r<s:\ \int_r^s \!\!\Psi^{\eps_k,*}_{u_{\eps_k}(t)}
    \big(B^{\eps_k}(t,u_{\eps_k}(t)){-}\xi_{\eps_k}(t)\big) \dd t \
    \to \int_r^s \!\!\Psi^{0,*}_{u(t)} \big(B^0(t,u_0(t)){-}\xi_0(t)\big) \dd t, 
 \end{align}
\end{subequations}
where $\xi_\eps(t)\in \pl\calE^\eps_t(u_\eps(t))$ for $\eps\in [0,1]$
and a.a.\ $t\in[0,T]$.   
\end{thm}

The proof of this result is contained in the following two Sections
\ref{su:EGC.CvgSubseq} and \ref{su:EGC.LimitPass}. However, we do not
give all the details and refer to the full proof of Theorem
\ref{th:MainExist} in Section \ref{se:ExistResult} for the details. 

\begin{rem}[Strong-weak closedness and
  $\Gamma$-convergence] \label{re:SWClosGamma}
It is a well-know fact that the strong-weak closedness in the limit
$\eps\to 0$ as assumed in \eqref{eq:cond.E.eps.5} often follows from
the $\Gamma$-convergence $\calE^\eps_t\Gto
\calE^0_t$. For the readers convenience we give the argument for the
convex case where $\pl\calE^\eps_t(u)$ is simply the convex
subdifferential, i.e.\ 
\[
\pl\calE^\eps_t(u)= \bigset{\xi \in V^*}{ \forall\, w\in V:\
  \calE^\eps_t(w)\geq \calE^\eps_t(u)+ \langle \xi,w{-}v\rangle }. 
\]
Thus, having a sequences $u_\eps \to u$ and $\xi_\eps \weak \xi_0$
with $ \xi_\eps \in \pl\calE^\eps_t(u_\eps)$ for $\eps >0$ and
$\calE^\eps_t(u_\eps)\to \ol e$, we can find, for each $w\in W$ a
recovery sequence $w_\eps \to w$ with $\calE^\eps_t(w_\eps)\to
\calE^0_t(w)$. Hence, we obtain
\[
\calE^\eps_t(w_\eps) \geq \calE^\eps_t(u_\eps) + \langle \xi_\eps ,
w_\eps {-}u_\eps\rangle \ \text{ for }\eps>0.
\]
Passing to the limit $\eps\to 0$ we obtain 
\begin{equation}
  \label{eq:w.u.xi}
  \calE^0_t(w) \geq \ol e + \langle \xi_0  , w{-}u_0\rangle, 
\end{equation}
where we used the strong convergence $w_\eps{-}u_\eps \to w{-}u_0$. 
By $\calE^\eps_t \Gto
\calE^0_t$ we already know $\calE^0_t(u_0)\leq \ol e$, but choosing
$w=u_0$ in \eqref{eq:w.u.xi} gives $\calE^0_t(u_0)= \ol e$ as
desired. With this, \eqref{eq:w.u.xi} immediately gives $\xi_0 \in
\pl\calE^0_t(u_0)$. 
\end{rem}

The above result is only one of many possible versions and several
generalizations are possible. For instance, we may combine time
discretization with time step $\tau\to 0$ with the limit $\eps \to
0$. More precisely, if we solve the time discretized problem (see Section 
\ref{su:TimeDiscret}) for $\PG^\eps$  with time step $\tau$ we obtain
an approximation $\wh U_{\tau_\eps}$. Then, it can be shown that these
approximations satisfy good a priori estimates and hence for every
sequence $(\tau_n,\eps_n)\to (0,0)$ there exists a subsequence and a
solution of $\PG^0$ such that the above convergences hold. We refer to
\cite[Thm.\,4.1]{MiRoSt08GLRR} or
\cite[Thm.\,3.12]{MiRoSa16BVSI} for results of this type.

\subsection{A priori estimates}
\label{su:EGC.Apriori}

The energy-dissipation principle states that every solution $u_\eps\in
\AC([0,T];V)$ for $\PG^\eps$, i.e.\ 
\eqref{eq:I.1} is satisfied, also satisfies the energy-dissipation
balance in the sense that there exists a measurable selection
$\xi_\eps:(0,T)\to V^*$ such that $\xi_\eps(t)\in
\pl\calE^\eps_t(u_\eps(t))$ a.e.\ in $(0,T)$ and that  
\begin{align}
 \nonumber
  &\calE^\eps_T(u_\eps(T))+ \int_0^T\!\! \Big(
  \Psi^\eps_{u_\eps(r)}(u'_\eps(r)) + \Psi^{\eps,*}_{u_\eps(r)}
  \big(B^\eps(r,u_\eps(r))-\xi_\eps(r)\big) \Big) \dd r 
\\
 &   \label{eq:EGC.EDB}
 = \calE^\eps_0(u_\eps(0)) + \int_0^T \!\! \Big(\pl_t \calE^\eps_r(u_\eps(r)) + 
  \big\langle B^\eps(t,u_\eps(r)),u'_\eps(t) \big\rangle \Big)\dd r. 
\end{align}
Estimating the last term via the \textsc{Young-Fenchel} inequality and
\eqref{eq:Beps.2} we obtain 
\[
\big\langle B^\eps(r,u_\eps(r)),u'_\eps(r) \big\rangle 
\leq 
\Psi^\eps_{u_\eps(r)}(u'_\eps(r)) + \Psi^{\eps,*}_{u_\eps(r)}
  \big(B^\eps(r,u_\eps(r))\big) \leq
  \Psi^\eps_{u_\eps(r)}(u'_\eps(r))+ C_B \calE^\eps_r(u_\eps(t))
\]
for the last term. Thus, the terms involving
$\Psi^\eps_{u_\eps(r)}(u'_\eps(r))$ and using $\Psi^{\eps,*}_{u}\geq
0$ and \eqref{eq:cond.E.eps.2} we arrive at 
\[
\calE^\eps_T(u_\eps(T))\leq \calE^\eps_0(u_\eps(0))+ \int_0^T
\big(C_T{+}C_B\big) \calE_r^\eps(u_\eps(r)) \dd r. 
\]
With $u_\eps(0)=u^0_\eps$ and the well-preparedness \eqref{eq:EGC.well}
the \textsc{Gronwall} lemma yields 
\[
G^\eps(u_\eps(t))\; \leq \;\calE_t^\eps(u_\eps(t)) \leq
2\calE^0_0(u_0^0) \,\ee^{(C_T+C_B) t} \leq \ol E:=
2\calE^0_0(u_0^0)\,\ee^{(C_T+C_B) T}.
\]
Thus, assumption \eqref{eq:cond.E.eps.2} guarantees that there exists
a compact set $K\Subset V$ such that $u_\eps(t)\in K$ for all
$(\eps,t)\in (0,1)\ti[0,T]$.  As $K\subset B_R(0)\subset V$ we can
apply the superlinearity \eqref{eq:Psi.eps.2} and the control \eqref{eq:Beps.2}
 of $B^\eps$ to estimate
\[
g_R\big(B^\eps(t,u_\eps(t))\big) \leq \Psi^{\eps,*}_{u_\eps(t)}
\big(B^\eps(t,u_\eps(t))\big) \leq C_B \calE^\eps_t(u_\eps(t)) \leq
C_B \ol E.
\]
This implies the boundedness of the non-variational perturbation,
viz.\
\begin{equation}
  \label{eq:EGC.B.bound}
  \exists\, R^*_B>0\ \forall\, (\eps,t)\in (0,1)\ti[0,T]: \quad 
 \| B^\eps(t,u_\eps(t))\|_{V^*}\leq R^*_B. 
\end{equation}

Inserting the bounds for $\calE^\eps_t(u_\eps(t))$ (and hence for
$\pl_t \calE^\eps_t(u_\eps(t))$) and for $B^\eps(t,u_\eps(t))$ into
\eqref{eq:EGC.EDB} we obtain 
\[
\int_0^T \!\! \Big(
  \Psi^\eps_{u_\eps(r)}(u'_\eps(r))- R^*_B\|u'_\eps(r)\|_V +
  \Psi^{\eps,*}_{u_\eps(r)} 
  \big(B^\eps(r,u_\eps(r))-\xi_\eps(r)\big) \Big) \dd r \leq C_E.
\]
Using that $\Psi^\eps$ and $\Psi^{\eps,*}$ are bounded from below by the
superlinear function $g_R$ (cf.\ \eqref{eq:Psi.eps.2}) and using
\eqref{eq:EGC.B.bound} again we arrive at
\begin{equation}
  \label{eq:EGC.gR.bounds}
   \exists\, C_\Psi>0\ \forall\, \eps\in (0,1]:\quad 
\int_0^T \big( g_R(\|u'_\eps(t)\|_V)+ g_R(\|\xi_\eps\|_{V^*}) \big)
\dd t \leq C_\Psi.
\end{equation}

\subsection{Convergent subsequences}
\label{su:EGC.CvgSubseq}

By \eqref{eq:EGC.gR.bounds} and the criterion of \textsc{de la
Vall\'ee-Poussin} for uniform integrability, the family
$u_\eps:[0,T]\to V$ is equi-continuous. As all values $u_\eps(t)$ lie
in the compact set $K$ the \textsc{Arzel\`a-Ascoli} theorem (e.g.\
\cite[Prop.\,3.3.1]{AmGiSa05GFMS}) gives
a subsequence $\eps_k\to 0$ such that the uniform convergence 
\eqref{eq:EGC.cvg.a}
holds. Moreover, \eqref{eq:EGC.gR.bounds} also implies weak
compactness, hence we may also assume 
$u'_{\eps_k} \weak u'_0$ in $\rmL^1(0,T;V)$, which is \eqref{eq:EGC.cvg.c}.

By the continuity \eqref{eq:Beps.1} we obtain convergence of the
non-variational terms, namely 
\begin{equation}
  \label{eq:EGC.cvg.B}
  \forall\, t\in [0,T]: \quad B^{\eps_k}(t,u_{\eps_k}(t)) \to
  B^0(t,u_0(t)) \text{ uniformly in } V^*. 
\end{equation}

Using the positivity of $\Psi^\eps$ and $\Psi^{\eps,*}$ we then obtain
that $\ol e^\eps:t\mapsto \calE^\eps_t(u_\eps(t))$ are uniformly
bounded in BV$([0,T])$, such that Helly's selection principle allows
to extract a subsequence (not relabeled) such that 
\begin{equation}
  \label{eq:EGC.E.cvg}
\forall\ t\in [0,T]:\quad \ol e^{\eps_k}(t) \to \ol e^0(t) \geq
\calE_t^0(u_0(t)),   
\end{equation}
where the last estimate follows from \eqref{eq:cond.E.epsGamma}. 

Again based on the superlinear bounds \eqref{eq:EGC.gR.bounds} we can
define extract further subsequence (not relabeled) such that $t\mapsto
(u'_{\eps_k}(t),\xi_{\eps_k}(t),\pl_t\calE^{\eps_k}_t(u_{\eps_k}(t)))$
generates a \textsc{Young} measure $\mathbold\mu=(\mu_t)_{t\in [0,T]} \in
\calY([0,T];V\ti V^*\ti \R)$ in the sense that
\begin{equation}
  \label{eq:EGC.F.mu}
  \int_0^T F\big(t,u'_{\eps_k}(t), \xi_{\eps_k}(t),
\pl_t\calE^\eps_t(u_\eps(t)) \big) \dd t  \to \int_0^T \int_{V\ti
  V^*\ti\R} F(t, v,\zeta,p)\dd \mu_t(v,\zeta,p) \dd t,
\end{equation}
for all continuous functions $F:[0,T]\ti V\ti
  V^*\ti\R \to \R$, where $V\ti V^*$ is equipped with the weak
  topology, with $F(t,v,\zeta,p)\leq C(1{+}\|v\|+\|\zeta\|_*)$. We
  refer to Appendix \ref{se:Appendix}.

\subsection{Limit passage and conclusion of the proof of Theorem \protect
  {\ref{th:EGC.main}}}  
\label{su:EGC.LimitPass}

We can now go back to the energy-dissipation balance
\eqref{eq:EGC.EDB} and pass to the limit $\eps_k\to 0$, where we
employ \textsc{Balder}'s lower semicontinuity result \cite{Bald84GALS}
for weakly normal integrands in the form of
\cite[Thm.\,4.3]{Stef08BEPD}, see Theorem \ref{th:A1}. The main point
here is that for $\bfalpha=(\alpha_1,\alpha_2,\alpha_3)\in
{[0,\infty)}^3$ the mappings
\[
F^\bfalpha_k:\;[0,T]\ti V\ti V^*\ti \R\to \R;\  (t,v,\zeta,p) \;\mapsto\;
\alpha_1\Psi_{u_{\eps_k}(t)}(v) +
\alpha_2\Psi^{\eps,*}_{u_{\eps_k}(t)}(\zeta) 
+ \alpha_3 p    ,
\]
satisfy a liminf estimate, namely 
\[
(v_k,\zeta_k,p_k)\weak (v,\zeta,p) \text{ in }V\ti V^* \ti \R \quad
\Longrightarrow \quad \liminf_{k\to \infty}
F^\bfalpha_k(t,v_k,\zeta_k,p_k) 
\geq F^\bfalpha_\infty(t,v,\zeta, p),
\]
where $F^\bfalpha_\infty(t,v,\zeta,p)= \alpha_1 \Psi^0_{u_0(t)}(v)
+ \alpha_2 \Psi^{0,*}_{u_0(t)}(\zeta) + \alpha_3 p$.  But the latter
liminf estimate 
follows easily from the \textsc{Mosco} convergence condition
\eqref{eq:Psi.eps.3}, because we already now $u_{\eps_k}\to u_0(t)$
and $\calE^{\eps_k}_t(u_{\eps_k}(t)) \leq \ol E$. In particular, we
obtain the three liminf estimates
\begin{subequations}
 \label{eq:EGC.liminf}
\begin{align}
  \label{eq:EGC.liminf.a}
  &\int_r^s \Psi^0_{u_0(t)} (u'_0(t)) \dd t \leq \liminf_{k\to \infty} 
   \int_r^s \Psi^{\eps_K}_{u_{\eps_K}(t)} (u'_{\eps_k}(t)) \dd t,
\\
  \label{eq:EGC.liminf.b}
  &\int_r^s \Psi^{0,*}_{u_0(t)}\big(
            B^{0_k}(t,u_0(t)){-}\xi_0(t))\big) \dd t
 \leq \liminf_{k\to \infty} 
   \int_r^s \Psi^{\eps_k,*}_{u_{\eps_k}(t)}\big(
  B^{\eps_k}(t,u_{\eps_k}(t)){-}\xi_{\eps_k}(t))\big) \dd t,
\\
  \label{eq:EGC.liminf.c}
& \int_r^s \pl_t \calE^0_t(u_0(t)) \dd t \leq \liminf_{k\to \infty} 
   \int_r^s \pl_t\calE^{\eps_k}_t(u_{\eps_k}(t)) \dd t,
\end{align}
\end{subequations}
where $0\leq r < s \leq T$ are arbitrary. 

Adding the three inequalities in \eqref{eq:EGC.liminf} and using the
limit $\ol e^0$ in \eqref{eq:EGC.E.cvg} we arrive at 
\begin{align}
 \nonumber
 &\ol e^0(T) + \int_0^T \!\!\int_{V\ti V^*\ti \R} \!\!\Big(\Psi^0_{u_0(t)}(v)+
  \Psi^{0,*}_{u_0(r)} \big(B^0(r,u_0(r))-\zeta\big) - p\Big) \dd
  \mu_t(v,\zeta,p) \dd t 
 \\
 \label{eq:EGC.est3}
 &\leq  \calE^0_0(u^0_0)+ \int_0^T \big\langle
  B^0(r,u_0(r)), u'_0(r) \big\rangle \dd r.
\end{align}
Here the convergence of the right-hand side follows from the
well-preparedness \eqref{eq:EGC.well} and the fact that the strong
$\rmL^{\infty}$ convergence \eqref{eq:EGC.cvg.B} and the weak convergence
\eqref{eq:EGC.cvg.c} imply the convergence of the integral. 
 
Now we exploit the main structural property of the \textsc{Young} measure
$\mathbold\mu$ which states that for a.a.\ $t\in [0,T]$ the supports
of $\mu_t$ lie in the set of accumulation points of defining
sequences. More, there is a null set $N\subset [0,T]$ (i.e.\ $|N|=0$)
such that 
\[
\forall\, t\in [0,T]{\setminus} N:\
 \mathrm{sppt}(\mu_t)\subset \mathrm{Li}(t):= \bigcap_{m=1}^\infty
\mathrm{clos_{weak}}\Big(
\bigset{(u'_{\eps_{k}}(t),\xi_{\eps_{k}}(t),\pl_t\calE^{\eps_{k}}_t(u_{\eps_{k}}(t)))
} { k\geq m} \Big) .
\]
Hence, the closedness condition \eqref{eq:cond.E.eps.5} guarantees that
\begin{align*}
&\forall\, t\in [0,T]{\setminus} N\ \forall\, (v,\zeta,p)\in
\mathrm{sppt}(\mu_t): \  
\zeta \in \pl\calE^0_t(u_0(t)), \ p\leq \pl_t\calE^0_t(u_0(t)),
\ \ol e^0(t)=\calE^0_t(u_0(t)).
\end{align*} 

We can now estimate further in \eqref{eq:EGC.est3}. By \eqref{eq:EGC.E.cvg} 
the first term $\ol e^0(T)$ is estimated from below by
$\calE^0_T(u_0(T))$. The term involving 
$\Psi^0_u(v)$ can be estimated by the convexity of $\Psi^0_u(\cdot)$
and the fact that $\mu_t$ is a probability measure with
$v$-expectation $u'_0$, i.e.\ 
\[
 u'_0(t)= \int_{V\ti V^*\ti\R}  v \dd \mu_t(v,\zeta,p).
\]
This follows simply by testing \eqref{eq:EGC.F.mu} by
$F(t,v,\zeta,p)=\langle \eta(t),v\rangle$ for all $\eta \in
\rmL^\infty(0,T;V^*)$. Thus, we have 
\[
\int_0^T \Psi_{u_0(t)}(u'_0(t)) \dd t \leq 
\int_0^T \int_{V\ti V^*\ti\R} \Psi_{u_0(t)}(v) \dd \mu_t(v,\zeta,p) \dd t,
\]

For the term involving 
$\Psi^{0,*}_u(v)$ we cannot apply Jensen's inequality as
$\pl\calE_t^0(u)$ may not be convex. Thus, for $t\in  [0,T]{\setminus}
N$ we select $\xi_0(t) \in \pl\calE_t^0(u_0(t))$ with 
\[
\Psi^{0,*}_{u_0(r)} \big(B^0(r,u_0(r))-\xi_0(t)\big) = \min\bigset{
\Psi^{0,*}_{u_0(r)} \big(B^0(r,u_0(r))-\zeta\big)}{ \zeta \in
\pl\calE_t^0(u_0(t))}. 
\]
Such a measurable selection exists, see Lemma \ref{le:A4} in Appendix
\ref{se:Appendix}. 

Finally using $p\leq \pl_t\calE^0_t(u_0(t))$ on Li$(t)$ the estimate 
\eqref{eq:EGC.est3} yields, for all $s\in (0,T]$,
\begin{align}
 \nonumber
 &\calE^0_s(u_0(s)) + \int_0^s \!\! \Big(\Psi^0_{u_0(t)}(u'_0(t)) +
  \Psi^{0,*}_{u_0(r)} \big(B^0(r,u_0(r))-\xi_0(t)\big) -
  \pl_t\calE^0_t(u_0(t)) \Big)  \dd t 
 \\
 \label{eq:EGC.est5}
 &\qquad \leq  \calE^0_0(u^0_0)+ \int_0^s \big\langle
  B^0(t,u_0(t)), u'_0(t) \big\rangle \dd t.
\end{align}
Moreover, by the \textsc{Fenchel-Young} inequality and the chain-rule
inequality \eqref{eq:cond.E.eps.4}, which 
is used for $\eps=0$ only, the 
left-hand side can be estimated from below via
\begin{align}
\nonumber
&\calE^0_s(u_0(s)) + \int_0^s \!\! \Big(\Psi^0_{u_0(t)}(u'_0(t)) +
  \Psi^{0,*}_{u_0(t)} \big(B^0(t,u_0(t))-\xi_0(t)\big) -
  \pl_t\calE^0_t(u_0(t)) \Big)  \dd t \\
\nonumber
&\overset{\text{FY}}{\geq} \calE^0_s(u_0(s)) + \int_0^s \!\!\Big(\langle
   B^0(t,u_0(t)){-}\xi_0(t) , u'_0(t)\rangle  -
  \pl_t\calE^0_t(u_0(t))  \Big) \dd t \\
\nonumber
&\overset{\text{chain}}\geq \calE^0_s(u_0(s)) + \int_0^s\!\! \Big(\langle
B^0(t,u_0(t)) , u'_0(t)\rangle -\frac{\rmd}{\rmd t}\big(
\calE_t^0(u_0(t))\big)  \Big) \dd t 
\\
\label{eq:EGC.est7}
&= \calE^0_0(u_0(0)) + \int_0^s \big\langle
  B^0(t,u_0(t)), u'_0(t) \big\rangle \dd t.
\end{align} 

Thus, we conclude that all inequalities in \eqref{eq:EGC.est5} and
\eqref{eq:EGC.est7} are equalities, which implies the the
\textsc{Fenchel-Young} estimate has to hold with equality a.e.\ in
$[0,T]$, which gives the desired differential inclusion $
B^0(t,u_0(t)) - \xi_0(t) \in \pl \Psi^0_{u_0(t)}(u'_0(t))$ or 
\[
 B^0(t,u_0(t)) \in \pl \Psi^0_{u_0(t)}(u'_0(t)) + \pl\calE^0_t(u_0(t))
 \qquad \text{a.e.\ in } [0,T].
\] 

Additionally, we observe that the liminf estimates 
\begin{align*}
\calE^0_t(u_0(t))
\leq \ol e^\infty(t)=\lim_{\eps_k\to 0} \calE^{\eps_k}_t(u_{\eps_k}(t))
\end{align*}
as well as the liminf estimates in \eqref{eq:EGC.liminf} are indeed
equalities as well. Thus, \eqref{eq:EGC.cvg.b}, \eqref{eq:EGC.cvg.d},
 and \eqref{eq:EGC.cvg.b} are established and the proof of Theorem
 \ref{th:EGC.main} is complete.

\subsection{Improved result for state-independent dissipation}
 \label{su:EGC.StateIndep}

The result on evolutionary $\Gamma$-convergence given in Theorem
\ref{th:EGC.main} has a rather strong assumption, namely the \textsc{Mosco}
convergence of $(\eps,u)\mapsto \Phi^\eps_u(\cdot)$ in the space $V$. 
This assumption is too strong for a series of important applications. 
For instance, for the parabolic equation
\[
\big( 2+ \cos(x_1/\eps)\big) u' = \Div\big( A(\tfrac1\eps x)\nabla
u\big) \text{ in } \Omega \subset \R^d, \qquad u=0 \text{ on } \pl\Omega,
\]
we may choose the gradient structure $(\bfQ,\calE^\eps,\Psi^\eps)$
with 
\[
V=\rmL^2(\Omega),\quad \calE^\eps(u)=\int_\Omega \frac12 \nabla u
\cdot A(\tfrac 1\eps x) \nabla u \dd x, \quad 
\Psi^\eps(v) = \int_\Omega \frac{2{+}\cos(x_1/\eps)}2\: v(x)^2 \dd x.
\]
However, $\Psi^\eps$ $\Gamma$-converges to $\Psi_\text{harm}$ in the
weak topology of $\rmL^2(\Omega)$ while it $\Gamma$-converges to
$\Psi_\text{arith}$ in the strong topology.

Here we want present a generalized version of \cite{LieRei15?HCHT}
where evolutionary $\Gamma$-convergence was established under the
weaker assumption $\Psi^\eps \Gto \Psi^0$, i.e.\ $\Gamma$-convergence
in the strong topology only.

If we inspect the proof in the previous subsection, then we see that
the weak $\Gamma$-convergence of $\Psi^\eps_{u_\eps}$ was used only
once, namely for deriving the liminf estimate
\eqref{eq:EGC.liminf.a}. The point is that we only derived the weak
convergence $u'_{\eps_k} \weak u'_0$ in $\rmL^1(0,T;V)$. However, the
``weak'' convergence may have two origins, namely first due to
oscillations in time and second due to weak convergence of
$u'_\eps(t) \weak u'_0(t)$ in $V$. The idea in \cite{LieRei15?HCHT} is
to consider piecewise affine interpolants $u_{\eps,\tau}$ of $u_\eps$
for fixed time steps $\tau>0$. This averages potential oscillations in
time as $u'_{\eps,\tau}$ is piecewise constant. Moreover, we can use
the strong convergence of $u_{\eps_k}(t)\to u_0(t)$ which implies that
$ u'_{\eps_k,\tau}(t)\to u'_{0,\tau}(t)$ in $V$ for a.a.\ $[t\in
[0,T]$. Finally, the limit $\tau\to 0$ is done after the limit $\eps_k\to
0$ is already performed. 
  
Our precise assumptions, which replace \eqref{eq:Psi.eps.3}, are the
following:
\begin{subequations}
\label{eq:Psi.strongG}
 \begin{align}
 \nonumber 
 &\text{{\bfseries Uniform
     continuity.} \quad For all } R>0\\
 \nonumber 
 &\quad \exists\; \text{modulus of continuity } \omega_R\  
   \forall \, \eps\in [0,1]\ \forall\, u_1,u_2 \text{ with }G^\eps(u_j)\leq R\\
 &\label{eq:Psi.strongG.a}\quad   \forall \, v\in V: \quad 
 \big|\Psi^\eps_{u_1}(v) - \Psi^\eps_{u_2}(v)\big| \leq \omega_R(\|u_1{-}u_2\|_V)
  g_R(\|v\|_V) ,\\
 \nonumber 
 &\text{{\bfseries Strong $\Gamma$-convergence.} \quad For all $R>0$
   we have}\\ 
 & \label{eq:Psi.strongG.b}
  \qquad u_\eps \to u_0\text{ and } \sup \calE^\eps_t(u_\eps)\leq R
  \quad \Longrightarrow \quad \Psi_{u_\eps}^\eps \Gto \Psi^0_{u_0},  
 \end{align}
\end{subequations}
where $g_R$ is the coercivity function defined in
\eqref{eq:Psi.eps.2}.

\begin{cor}[Strong $\Gamma$-convergence for $\calE^\eps$ and
  $\Psi^\eps$] 
\label{co:StrongGa}
All results of Theorem \ref{th:EGC.main} remain true if assumption
\eqref{eq:Psi.eps.3} is replaced by \eqref{eq:Psi.strongG}. 
\end{cor}
\begin{proof}
  To start with, we recall that the strong $\Gamma$-convergence of
  \eqref{eq:Psi.strongG.b} implies the weak $\Gamma$-convergence of
  the \textsc{Legendre-Fenchel} dual, i.e.\
  $\Psi^{\eps,*}_{u_\eps}\overset{\Gamma}\weak \Psi^{0,*}_u$, see
  \eqref{eq:Gcvg.sw*}.  Thus, the liminf estimate
  \eqref{eq:EGC.liminf.b} follows exactly as above.
 
Thus, it remains to find a new proof for the liminf estimate
\eqref{eq:EGC.liminf.a}. Using the notation
\[
J^\eps(u,v) := \int_0^T \Psi^\eps_{u(t)}\big( v(t)\big) \dd t
\]
we have to show $\liminf_{k\to \infty}
J^{\eps_k}(u_{\eps_k},u'_{\eps_k})\geq J^0(u_0,u'_0)$, where our
sequence $(u_{\eps_k})_k$ satisfies
\begin{align*}
\text{(a) } \| u_{\eps_k}{-}u_0\|_{\rmC^0([0,T];V)} \to 0, 
 \text{ (b) } \|u'_{\eps_k}{-}u'_0\|_{\rmL^1(0,T;V)} \to 0, 
\text{ (c) }  \int_0^T \!\!g_R\big(\| u'_{\eps_k}(t)\|\big) \dd t \leq C_g,
\end{align*}
where $R\geq \sup\set{\|u_{\eps_k}\|_\infty}{ k\in \N}$. 

For time steps $\tau=T/N>0$ with $N\in \N$ we define piecewise
constant and piecewise affine interpolants $\ol u^\tau_{\eps_k}$ and $\wh
u^\tau_{\eps_k}$ as in \eqref{eq:Approx.tau}. By the uniform convergence (a) we
have equi-continuity of the sequence $(u_{\eps_k})_k$, and hence
\[
\mu_\tau := \sup\set{\| u_{\eps_k}-\ol u_{\eps_k}\|_{\rmC^0([0,T];V)}
}{k \in \N} \ \to \ 0 \quad \text{for }\tau\to 0.
\]
With \eqref{eq:Psi.strongG.a} and (c) we obtain the lower bound 
\begin{align*}
J^{\eps_k}(u_{\eps_k},u'_{\eps_k}) &\geq
J^{\eps_k}(\ol u^\tau_{\eps_k},  u'_{\eps_k}) - \int_0^T\!\!
\omega_R\big(\|u_{\eps_k}\!{-} \ol u^\tau_{\eps_k}\|\big)
g_R\big(\|u'_{\eps_k}\|\big) \dd t \geq J^{\eps_k}(\ol
u^\tau_{\eps_k},  u'_{\eps_k}) - \omega_R(\mu_\tau) C_g.
\end{align*}
On the intervals $((n{-}1)\tau,n\tau)$ the integrand $\Psi^{\eps_k}_{\ol
u^\tau_{\eps_k}(t)}(\cdot)$ is independent of $t$ and convex. Hence,
we can apply Jensen's inequality and replace
$v_k(t)=u'_{\eps_k}(\cdot)$ by its average over this interval, which
is exactly 
\[
\frac1\tau \int_{(n-1)\tau}^{n\tau} u'_{\eps_k}(r) \dd r =
\frac1\tau\big(u_{\eps_k}(n\tau)-u_{\eps_k}((n{-}1)\tau)\big)=
\wh u^{\tau}_{\eps_k}\!{}'(t) \quad \text{for } t\in ((n{-}1)\tau,n\tau).
\]
Thus, we have the lower bound $J^{\eps_k}(u_{\eps_k},u'_{\eps_k}) \geq
J^{\eps_k}(\ol u^\tau_{\eps_k},\wh u^{\tau}_{\eps_k}\!{}')-
\omega_R(\mu_\tau)C_g$.\medskip

For $k\to \infty$ we have $ \ol u_{\eps_k}\to \ol u_0$ in $V$ and $
\wh u^\tau_{\eps_k}\!{}'\to \wh u^\tau_0{}'$ in $V$ a.e.\ in
$[0,T]$. Hence, we can exploit the liminf estimate of the strong
$\Gamma$-convergence $\Psi^\eps_{u_\eps} \Gto \Psi^0_{u_0}$. \textsc{Fatou}'s
lemma leads to
\begin{align*}
&\liminf_{k\to \infty} J^{\eps_k}(u_{\eps_k},u'_{\eps_k})\geq 
\liminf_{k\to \infty} J^{\eps_k}(\ol u^\tau_{\eps_k},\wh u^{\tau}_{\eps_k}\!{}')-
\omega_R(\mu_\tau)C_g\\
&\overset{\text{Fatou}}\geq  J^0(\ol u^\tau_0,\wh u^{\tau}_0{}')-
\omega_R(\mu_\tau)C_g \  \geq \ J^0(u_0,\wh
u^{\tau}_0{}')-2\omega_R(\mu_\tau)C_g ,
\end{align*}
where we used $\|u_0-\ol u^\tau_0\|_\infty \leq \mu_\tau$ for the last
step. 

Thus, using $\omega_R(\mu_\tau)\to 0$ for $\tau\to 0$ it remains to
show that $L:=\liminf_{\tau \to 0} J^0(u_0, \wh u^{\tau}_0{}') \geq
J^0(u_0,u'_0)$. Choose a subsequence $\tau_m$ such that $ J^0(u_0, 
\wh u^{\tau_m}_0{}')\to L$. We now use the well-known fact that 
$\wh u^{\tau_m}_0{}' \to u'_0$ in $\rmL^1(0,T;V)$, which implies that 
there exists a further subsequence (not relabeled) such that 
$\wh u^{\tau_m}_0{}'(t) \to u'_0(t)$ in $V$ a.e.\ in $[0,T]$. Moreover,
$\Psi^0_{u_0(t)}(\cdot):V\to [0,\infty)$ is continuous, because it is
convex and bounded from above by the \textsc{Legendre-Fenchel} dual of $\xi
\mapsto g_R(\|\xi\|_{V^*})$. This gives $\Psi^0_{u_0(t)}(\wh
u^{\tau_m}_0{}'(t)) \to \Psi^0_{u_0(t)}(u'_0(t))$ a.e.\ in $[0,T]$,
and \textsc{Fatou}'s lemma implies $L=\liminf_{m\to \infty} J^0(u_0, \wh
u^{\tau_m}_0{}') \geq J^0(u_0,u_0')$ as desired.
 
Altogether we have established $\liminf_{k\to \infty}
J^{\eps_k}(u_{\eps_k},u'_{\eps_k})\geq J^0(u_0,u'_0)$, and thus
Corollary \ref{co:StrongGa} is proved. 
\end{proof}

\section{Homogenization of reaction-diffusion systems}
\label{se:Example}

In this section we provide a nontrivial example that highlights the
applicability of our abstract existence theory as well as the theory
of evolutionary $\Gamma$-convergence. We refer to
\cite{MiReTh14TSHN, Reic16EEEE, Reic17CECI} and the references therein
for general homogenization results that are typically for semilinear
systems where the leading order terms are decoupled.  
Our example of a reaction
diffusion system is a general quasilinear parabolic system, where the
leading terms may be coupled but need to have a variational
structure. 

Our system for the vector $u(t,x)\in \R^I$ reads as follows:
\begin{align}
\nonumber
 \hspace*{1em}
 A^\eps(x,u(t,x)) \pl_t u(t,x) &= \Div \Big(\pl_{\nabla u}
  F^\eps \!\,\big(x,u(t,x),\nabla u(t,x)\big) \Big) \\
\label{eq:ExaPDE}
 & \quad - \pl_u F^\eps\!\,(x,u(t,x),\nabla
  u(t,x)) + b^\eps(x,t,u(t,x)) &\text{in }&\Omega,\\
\nonumber
 0&=\pl_{\nabla u}
  F^\eps\!\,\big(x,u(t,x),\nabla u(t,x)\big)\nu(x)&\text{on }&\pl\Omega.
\hspace*{1em}
\end{align}
Generally we assume that $\Omega\subset \R^d$ is a
bounded domain with Lipschitz boundary $\pl\Omega$. For simplicity,
we have imposed Neumann boundary conditions only, but more general
conditions including Dirichlet or Robin boundary conditions could be
used as well.  
 
We first summarize the needed assumptions on the functions $A^\eps$,
$F^\eps$, and $b^\eps$, then show that these assumptions imply the
once needed for the existence theory in Section \ref{se:ExistResult},
and finally discuss under which conditions we have evolutionary
$\Gamma$-convergence for $\eps\to 0$.

\subsection{The existence result}
\label{su:Exa.Exist}

For the matrix $A^\eps(x,u)\in \R^{I\ti I}_\text{sym}:=\set{A\in
  \R^{I\ti I}}{A=A^\top}$  we make the assumption
\begin{subequations}
\label{eq:Exa.Ass}
\begin{align}
 \label{eq:Exa.Ass.A1} 
 &\forall\,\eps \in [0,1]:\quad  A^\eps:\Omega\ti \R^I \to
 \R_\text{sym}^{I\ti I}\text{ is a \textsc{Carath\'eodory} function}, \\ 
 \label{eq:Exa.Ass.A2} 
 &\exists\, C_A>0\  \forall\,\eps \in [0,1]\ \forall\, x\in
 \Omega\ \forall\, u,v\in  \R^I:\ \  \frac1{C_A} |v|^2 \leq \langle
 A^\eps(x,u)v , v\rangle \leq C_A|v|^2. 
\end{align}
Here $G:\Omega\ti \R^M\to \R^N$ is called a \textsc{Carath\'eodory} function,
if $x\mapsto G(x,z) $ is measurable for all $z\in \R^m$ and  $z
\mapsto G(x,z)$ is continuous for a.a.\ $x\in \Omega$. 

For simplicity, we will assume that the functions
$F^\eps(x,\cdot,\cdot)$ are convex, but much weaker conditions would be
possible (e.g.\ $\lambda$-convexity in $u$ or poly-convexity in $U=\nabla u$). 
\begin{align}
 \label{eq:Exa.Ass.F1} 
 &\forall\,\eps \in [0,1]:\quad  F^\eps:\Omega\times (\R^I\ti \R^{I\ti d}) \to
 \R\text{ is a \textsc{Carath\'eodory} function}, \\ 
 \label{eq:Exa.Ass.F2} 
 &\forall\,\eps \in [0,1]\ \forall_\text{a.a.}x\in \Omega:\quad
 F^\eps(x,\cdot,\cdot):\R^I\ti \R^{I\ti d} \to
 \R\text{ is convex} , \\
 \nonumber 
 &\exists\, C_F>0 \ \exists\, p,q>1\  \forall\,\eps \in [0,1]\ \forall\,
 (x,u,U) \in \Omega\ti \R^I\ti \R^{I\ti d}:\\
 &\label{eq:Exa.Ass.F3} \qquad  F^\eps (x,u,U)\geq
  C_F\big(1+|u|^q+|U|^p\big). 
\end{align}
For the non-gradient terms $b^\eps$ we impose the following conditions:
\begin{align}
 \label{eq:Exa.Ass.B1} 
 &\forall\,\eps \in [0,1]:\quad  b^\eps:\Omega\times ([0,T]\ti
 \R^I) \to \R^I \text{ is a \textsc{Carath\'eodory} function}, \\ 
 \nonumber
   &\exists\, h\in \rmL^2(\Omega),\;  C_B>0, \;  r>1\
 \forall\,(\eps,t,x,u) \in [0,1]\ti \Omega\ti [0,T]\ti\R^I:\\ 
 \label{eq:Exa.Ass.B2} & \qquad |b^\eps(x,t,u)| \leq h(x)+ C_B |u|^r.
\end{align}
\end{subequations}

We choose basic space $V= \rmL^2(\Omega;\R^I)$,
the energy functionals 
\[
 \calE^\eps(u)= \left\{ \ba{cl}\ds\int_\Omega F^\eps(x,u(x),\nabla
  u(x))\dd x &\text{for }u \in \rmW^{1,p}(\Omega;\R^I),\\ 
   \infty& \text{otherwise,}\ea\right. 
\]
and the dissipation potentials 
\[
\Psi^\eps_u(v):= \int_\Omega \frac12 \langle A^\eps(x,u(x)) v(x) ,
v(x)\rangle \dd x.
\]
Thus, the perturbed gradient systems $\PG^\eps=(V,\calE^\eps,
\Psi^\eps,b^\eps)$ is fully specified, and we want to apply our
abstract theory. Before doing so, we note that in the our conditions
the exponent $q$ appears three times: (i) the first relation in
\eqref{eq:Coeff.Rel} below implies $\rmW^{1,p}(\Omega) \subset
\rmL^q(\Omega)$, (ii) the coercivity \eqref{eq:Exa.Ass.F3} of $F$ asks
for the lower bound $C_F|u|^q$, and (iii) the second relation in
\eqref{eq:Coeff.Rel} says that $B(\cdot , u(\cdot))$ is controlled by
$C(1{+}\|u\|_q^q)$.

\begin{pro}\label{pr:Exa.Exist} Let the functions $A^\eps$, $F^\eps$,
  and $b^\eps$ satisfy the conditions \eqref{eq:Exa.Ass}, where the
  coefficients $p$, $q$, and $r$ satisfy the relations
  \begin{equation}
   \label{eq:Coeff.Rel}
    1-\frac dp> -\frac{d}{q} \quad \text{ and } \quad q \geq 2r. 
\end{equation}

Then, for each initial condition
  $u^0_\eps\in \rmL^2(\Omega;\R^I)$ with $\calE^\eps(u_\eps^0)<0$ there is a
  solution $u_\eps:[0,T]\to \rmL^2(\Omega;\R^I)$ of \eqref{eq:ExaPDE} 
 such that $u_\eps \in \rmH^1(0,T;\rmL^2(\Omega)) \cap
 \rmC^0_\text{weak}([0,T]; \rmW^{1,p}(\Omega))$.  
\end{pro}

The proof is a consequence of our abstract existence result in Theorem
\ref{th:MainExist}.
We easily find the \textsc{Legendre-Fenchel} dual $\Psi^{\eps,*}_u(\xi)=
\int_\Omega \frac12 \langle \xi(x),(A^\eps(x,u(x)))^{-1} \xi(x)
\rangle \dd x$.  Clearly, \ref{eq:Psi.1} holds and we have the equi-coercivities
\[
\Psi^\eps_u(v)\geq \frac1{2C_A}\|v\|^2_V \quad \text{and} \quad
\Psi^{\eps,*}_u(\xi)\geq \frac1{2C_A}\|\xi\|^2_{V^*},
\] 
which imply the desired superlinearities \ref{eq:Psi.2}. 
Finally, the \textsc{Mosco} convergence
$\Psi^\eps_{u_n}\Mto \Psi^\eps_u$ (here $\eps>0$ is still fixed)
follows since $u_n\to u$ in $V$ implies that $A^\eps(\cdot,u_n(\cdot))
\to A^\eps(\cdot,u(\cdot))$ a.e.\ in $\Omega$ along suitable
subsequences. To see that this is sufficient for
\textsc{Mosco}-convergence, we use the \textsc{Moreau-Yosida}
regularizations
\[
\Psi^{\eps,\lambda}_{u} (v):= \inf\Bigset{
  \Psi^\eps_{u_n}(w)+\frac\lambda{2}\|w{-}v\|_{\rmL^2}^2}{ w\in
  \rmL^2(\Omega;\R^I)} 
\]
where $\lambda>0$. It is easy to see that $\Psi^{\eps;\lambda}_u$ is
still quadratic, but now with the matrix $\lambda
A^\eps(A^\eps{+}\lambda I)^{-1}$. By \cite[Thm.\,3.26]{Atto84VCFO} we
have $\Psi^\eps_{u_n}\Mto \Psi^\eps_u$ if and only if for all $v\in
V=\rmL^2(\Omega;\R^I)$ and all $\lambda>0$ we have the pointwise
convergence
$\Psi^{\eps,\lambda}_{u_n}(v)\to\Psi^{\eps,\lambda}_u(v)$. But this
follows immediately by the boundedness of $A^\eps$ and
\textsc{Lebesgue}'s dominated convergence theorem.  Hence,
\ref{eq:Psi.3} is shown as well.

The energy functionals $\calE^\eps$ are convex and independent of
time. Hence \ref{eq:cond.E.1} and \ref{eq:cond.E.3} hold trivially. By the coercivity of
$F^\eps$ we obtain the coercivity of $\calE^\eps$, namely 
\begin{equation}
  \label{eq:E.coercive}
  \calE^\eps(u) \geq \int_\Omega C_F\big(1 +|u|^q+|\nabla u|^p
\big) \dd x \geq \wt c \|u\|_{\rmW^{1,p}}^{\min\{p,q\}}
-\wt C,  
\end{equation}
such that sublevels are bounded in $\rmW^{1,p}(\Omega;\R^I)$. Because
this space is compactly embedded in $V=\rmL^2(\Omega;\R^I)$ by
assumption \eqref{eq:Coeff.Rel}, we conclude that \ref{eq:cond.E.2}
holds.  The chain rule \ref{eq:cond.E.4} and the weak-strong
closedness of the \textsc{Fr\'echet} subdifferential (which is the
same as the convex subdifferential) follows by convexity, see Remark
\ref{re:SWClosGamma} or \cite{MiRoSa13NADN}.

We now set $B^\eps(t,u)(x)=b^\eps(x,t,u(x))$ and obtain the continuity
\ref{eq:B.1} simply from the continuity of $b^\eps(x,\cdot,\cdot)$ and
$2r\leq q$. The energy control \ref{eq:B.2} follows from
\eqref{eq:Exa.Ass.F3} and the second condition in
\eqref{eq:Coeff.Rel}.  Thus, all the abstract assumptions of Theorem
\ref{th:MainExist} are established, and Proposition \ref{pr:Exa.Exist}
is established.

\subsection{The homogenization result}
\label{su:Exa.Homog}

We want to apply the evolutionary $\Gamma$-convergence of Section
\ref{se:EGC} for homogenization, i.e.\ we assume that the
$x$-dependence of $A^\eps$, $F^\eps$, and $b^\eps$ is of oscillatory
type, namely 
\begin{equation}
 \label{eq:Coeff.2s}
  A^\eps (x,u)=\bbA(\inveps x, u),\quad F^\eps(x,u,U)=\bbF(\inveps x, u,
  U), \quad b^\eps(x,t,u)= \bbB(\inveps x , u),
\end{equation}
where the functions $\bbA$, $\bbF$, and $\bbB$ are assumed to be
1-periodic in all directions, i.e.\ $\bbG(y{+}k)=\bbG(y)$ for all
$y\in \R^d$ and $k\in \Z^d$.  

For the quadratic dissipation potentials $\Psi^\eps_u$ we have the
following $\Gamma$-convergences:
\begin{equation}
  \label{eq:HomPsiGaCvg}
 (\eps_n, u_n) \to (0,u) \in \R \ti \rmL^2(\Omega;\R^n) \ 
\Longrightarrow \ \Big(\Psi^{\eps_n}_{u_n} \Gweak \Psi^\text{harm}_u
\text{ and } \Psi^{\eps_n}_{u_n} \Gto \Psi^\text{aver}_u \Big), 
\end{equation}
where the harmonic-mean functional $\Psi^\text{harm}_u $ and the
average functional $\Psi^\text{aver}_u$ are defined via
\begin{align*}
\Psi^\text{harm}_u(v) &=\int_\Omega \frac12\langle A^\text{harm}(u(x))
v(x),v(x)\rangle \dd x &\text{with }&
A^\text{harm}(u)^{-1}=\int_{(0,1)^d} \bbA(y,u)^{-1}\dd y,\\
\Psi^\text{aver}_u(v) &=\int_\Omega \frac12\langle A^\text{aver}(u(x))
v(x),v(x)\rangle \dd x &\text{with }&
A^\text{aver}(u)=\int_{(0,1)^d} \bbA(y,u)\dd y.
\end{align*} 
The strong $\Gamma$-convergence $\Psi^{\eps_n}_{u_n} \Gto
\Psi^\text{aver}_u$ follows simply from the pointwise convergence
$\Psi^{\eps_n}_{u_n}(v) \to \Psi^\text{aver}_u(v)$ for all $v$ and the
equi-\textsc{Lipschitz} continuity. The weak $\Gamma$-convergence
$\Psi^{\eps_n}_{u_n} \Gweak \Psi^\text{aver}_u$ follows by
\eqref{eq:Gcvg.sw*} and \textsc{Legendre-Fenchel} transform as
$\Psi_u^{\eps,*}$ is given in terms of $(A^\eps)^{-1}$, see also
\cite[Exa.\,2.36]{Brai02GCB}.

In particular, we see that \textsc{Mosco} convergence only holds for the case
that the harmonic and the arithmetic mean are equal, which means that 
$\bbA(y,u)$ has to be independent of $y$. 

For the energy functional $\calE^\eps$ we can rely on the general
theory of homogenization as surveyed in
\cite{Brai06HGC}. Using the uniform coercivity
\eqref{eq:E.coercive} we obtain weak $\Gamma$-convergence in
$\rmW^{1,p}(\Omega;\R^I)$ and, by the compact embedding, strong
$\Gamma$-convergence in $V=\rmL^2(\Omega;\R^I)$ towards the limit
\begin{align*}
&\calE^0(u)=\int_\Omega F^\text{hom}(u(x),\nabla u(x))\dd x \text{ with }\\
&F^\text{hom}(u,U):= \min\Bigset{\int_{(0,1)^d} \bbF(y,u,U{+}\nabla\Phi(y)) \dd
  y}{ \Phi\in \rmW^{1,p}_\text{per}((0,1)^d;\R^I)},
\end{align*} 
see \cite[Thm.\,5.1, pp.\,135]{Brai06HGC}.  Of course, $\calE^0:V\to
[0,\infty]$ is a again a convex and lower semicontinuous functional.
Finally, setting 
\[
B^0(t,u): x \mapsto b^\text{aver}(t,u(x)) \quad \text{ with }
b^\text{aver}(t,u)=\int_{(0,1)^d} \bbB(y,u)\dd y   
\]
we obtain the desired convergence $B^{\eps_n}(t_n,u_n)\to B^0(t,u)$ if
$(\eps_n,t_n,u_n)\to (0,t,u)$ in $[0,1]\ti [0,T]\ti V$. 

Hence, we see that Theorem \ref{th:EGC.main}, which is the
main result on evolutionary $\Gamma$-convergence, is only applicable
if we have the \textsc{Mosco} convergence $\Psi^{\eps_k}_{u_k}\Mto \Psi^0_u$, which
means  $\Psi^\text{harm}_u = \Psi^\text{aver}_u$. Thus, we need to
assume that $\bbA(y,u)$ does not depend on the microscopic periodicity
variable $y \in \R^d/_{\Z^d}$. In summary we obtain the following
result. 

\begin{thm}[Homogenization I] \label{th:Hom.I}
Consider the perturbed gradient system $\PG^\eps=(\rmL^2(\Omega;\R^I),
\calE^\eps, \Psi^\eps, b^\eps)$ be 
given as above. Assume that \eqref{eq:Exa.Ass} holds and that
\eqref{eq:Coeff.2s} holds with $\bbA$ independent of the variable $y=\inveps 
x$, the we have evolutionary $\Gamma$-convergence in the sense of
Theorem \ref{th:EGC.main} to the perturbed gradient system $(\rmL^2(\Omega;\R^I), \calE^0,
\Psi^\text{aver}, b^\text{aver})$, i.e.\ solutions $u_\eps$ of the
reaction-diffusion system \eqref{eq:ExaPDE} converge to solutions of
the homogenized system 
\begin{align}
\nonumber
 \hspace*{1em}
 A^\text{aver}(u) \pl_t u &= \Div \big(\pl_{\nabla u}
  F^\text{hom} \!\,(u,\nabla u) \big) 
- \pl_u F^\text{hom}\!\,(u,\nabla u) + 
     b^\text{aver}(t,u) &\text{in }&\Omega,\\
\label{eq:ExaPDE.eff}
 0&=\pl_{\nabla u}
  F^\text{hom}\!\,(u,\nabla u)\nu &\text{on }&\pl\Omega.
\hspace*{1em}
\end{align}
\end{thm}

The case where $\bbA(y,u)$ depends on $y\in \R^d/_{/\Z^d}$ is more
difficult. Under additional assumptions we will be able to use the
improved theory developed in Corollary \ref{co:StrongGa}, as we can
use $\Psi^{\eps_k}_{u_k}\Gto \Psi^\text{aver}_u$, which gives
assumption \eqref{eq:Psi.strongG.b}. However, we need to
establish the uniform continuity \eqref{eq:Psi.strongG.a}. For this we
note that $G^\eps(u_j)\leq R$ implies $\| u_j\|_{\rmW^{1,p}}\leq
C_R$. Now, assuming $p>d$ we first observe $\|u_j\|_{\rmL^\infty}\leq
\wt C_R<\infty$, and a Gagliardo-Nirenberg estimate yields  
\[
\|u_1{-}u_2\|_{\rmL^\infty} \leq C_{\rmG\rmN}
\|u_1{-}u_2\|_{\rmL^2}^\theta \|u_1{-}u_2\|^{1-\theta}_{\rmW^{1,p}}
\leq  C_{\rmG\rmN} (2C_R)^{1-\theta} \|u_1{-}u_2\|_{\rmL^2}^\theta.
\]
Now assuming the uniform continuity 
\begin{equation}
  \label{eq:bbAUnifCont}
\begin{aligned}
&\forall\, \rho>0\ \exists\text{ modulus of contin.\,} \omega_\rho\
\forall\, y\in (0,1)^d\ \forall\, u_j\in B_\rho(0)\subset \R^I:\\
&\qquad 
\big|  \bbA(y,u_1)-\bbA(y,u_2)\big| \leq \omega_\rho(|u_1{-}u_2|) ,
\end{aligned}
\end{equation}
we can estimate the difference $ \Psi^\eps_{u_1}(v)-
\Psi^\eps_{u_2}(v)$ of the dissipation potentials pointwise under the
integral and obtain
\[
\forall\, u_j\in V \text{ with } G^\eps(u_j)\leq R: \quad 
\big|\Psi^\eps_{u_1}(v)- \Psi^\eps_{u_2}(v)\big| \leq \omega_{\wh
  C_R} \big( C_{\rmG\rmN} (2C_R)^{1-\theta}
\|u_1{-}u_2\|_{\rmL^2}^\theta\big) \| v\|_{\rmL^2}^2. 
\]
This is exactly the desired uniform continuity \eqref{eq:Psi.strongG.a}.
Thus, Corollary \ref{co:StrongGa} is applicable under the additional
assumption that $p>d$ and that \eqref{eq:bbAUnifCont} holds, which
gives our second homogenization result, where $\bbA$ now may depend
periodically on $y=\inveps x$.

\begin{thm}[Homogenization II] \label{th:Hom.II}
Consider the perturbed gradient systems $(\rmL^2(\Omega;\R^I), \calE^\eps, \Psi^\eps, b^\eps)$
given as above. Assume that \eqref{eq:Exa.Ass} holds with $p>d$ and that
\eqref{eq:Coeff.2s} together with \eqref{eq:bbAUnifCont}. Then all the
conclusions of Theorem \ref{th:Hom.I} remain true. 
\end{thm}

Indeed, we conjecture that these two additional conditions (either
$\bbA$ independent of $y=\inveps x$ or \eqref{eq:bbAUnifCont}) are not
really necessary. Using two-scale unfolding as in \cite{MiReTh14TSHN,
  Reic16EEEE, Reic17CECI} and   
a suitable version of \textsc{Ioffe}'s theorem it
should be possible to prove the fundamental liminf estimate 
\[
\int_0^T \Psi^\text{aver}_{u(t)}(u'(t))\dd t \leq \liminf_{k\to
  \infty} \int_0^T \Psi^{\eps_k}_{u_{\eps_k}(t)}(u'_{\eps_k}(t))\dd t
\]
in much more general cases. 
\EEE

\appendix
\section{Appendix}
\label{se:Appendix}

In this section, we provide some tools on parametrized \textsc{Young}
measures which we made use in the last section. First, we give some
notions related to \textsc{Young} measures and which was originally
introduced by \textsc{Balder} \cite{Bald84GALS}.

In the following, let $\mathcal{V}$ be a reflexive (separable)
\textsc{Banach} space. In fact, we employed the following results to
the separable and reflexive \textsc{Banach} space $\mathcal{V}=V\times
V^*\times R$ endowed with the product topology. Further, let for an
interval $\mathscr{L}_{(0,T)}$ be the \textsc{Lebesgue}
$\sigma$-algebra of $(0,T)$ and let $\mathscr{B}(\mathcal{V})$ be the
\textsc{Borel} $\sigma$-algebra of $\mathcal{V}$.

Then, we say that a $\mathscr{L}_{(0,T)} \otimes
\mathscr{B}(\mathcal{V})$-measurable function $\mathcal{H}(0,T)\times
\mathcal{V}\rightarrow (-\infty,+\infty]$ is a \textit{weakly-normal
  integrand} if for almost every $t\in(0,T)$ the map $w\mapsto
\mathcal{H}(t,w)$ is sequentially lower semicontinuous with respect to
the weak topology of $\mathcal{V}$.  
\medskip

Furthermore, we denote by $\mathscr{M}(0,T;\mathcal{V})$ the set of
all $\mathscr{L}_{(0,T)}$-measurable functions $y:(0,T)\rightarrow
\mathcal{V}$. Then, a sequence $(w_n)_{n\in\mathbb{N}}\subset
\mathscr{M}(0,T;\mathcal{V})$ is said to be \textit{weakly-tight} if
there exists a weakly-normal integrand $\mathcal{H}:(0,T)\times
\mathcal{V}\rightarrow (-\infty,+\infty]$ such that the map $w\mapsto
\mathcal{H}(t,w)$ has weakly compact sublevels in $\mathcal{V}$ for
a.a. $t\in(0,T)$, and there holds
\begin{align}
\label{eq:A1}
\sup_{n\in \mathbb{N}} \int_0^T \mathcal{H}(t,w_n(t))\dd r <+\infty. 
\end{align} 
Finally, a family $\mathbold{\mu}=(\mu_t)_{t\in (0,T)}$ of
\textsc{Borel} probability measures on $\mathcal{V}$ is called
\textsc{Young} measure if on $(0,T)$ the map $t\mapsto \mu_t(B)$ is
$\mathscr{L}_{(0,T)}$-measurable for all $B\in
\mathscr{B}(\mathcal{V})$. With $\mathscr{Y}(0,T;\mathcal{V})$ we
denote the set of all \textsc{Young} measures in $\mathcal{V}$.

\begin{thm}\label{th:A1}
  Let $\mathcal{H}_n,\mathcal{H}: (0,T)\times \mathcal{V}\rightarrow
  (-\infty,+\infty]$ be for all $n\in \mathbb{N}$ weakly normal
  integrand such that for all $w\in \mathcal{V}$ and for almost every
  $t\in(0,T)$ we have
  \begin{align}
    \label{eq:A2}
   \mathcal{H}(t,w)\leq \inf \lbrace \liminf_{n\rightarrow \infty}
   \mathcal{H}_n  (t,w_n)  \mid w_n\rightharpoonup w \quad
   \text{in } \mathcal{V} \rbrace. 
  \end{align} 
  Let $(w_n)_{n\in\mathbb{N}}\subset \mathscr{M}(0,T;\mathcal{V})$ be
  a weakly-tight sequence. Then, there exists a subsequence
  $(w_{n_k})_{k\in \mathbb{N}}$ and a \textsc{Young} measure
  $\mathbold{\mu}=(\mu_t)_{t\in(0,T)}$ such that for almost every
  $t\in(0,T)$ we have 
  \begin{align}
   \label{eq:A3}
   \mathrm{sppt}(\mu_t)\subset \mathrm{Li}(t):=
   \bigcap_{p=1}^{\infty}\mathrm{clos_{weak}} \big( \lbrace w_{n_k}(t)
   \mid k\geq p \rbrace \big), 
  \end{align} 
  i.e.\ $\mu_t$ is concentrated on the set of all limit points of the
  sequence $(w_{n_k}(t))_{k\in \mathbb{N}}$ with respect to the weak
  topology $W$ of $\mathcal{V}$, where $\overline{A}^{W}$ denotes the
  weak closure of a subset $A\subset\mathcal{V}$, and if the sequence
  $t \mapsto \mathcal{H}^-_n(t,w_{n_k}(t)):=\max \lbrace
  -\mathcal{H}_n(t,w_{n_k}(t)),0\rbrace$ is uniformly integrable,
  there holds
  \begin{align}
   \label{eq:A4}
   \int_0^T \int_{\mathcal{V}} \mathcal{H}(t,w)\dd \mu_t(w)\dd t\leq
   \liminf_{k\rightarrow}\int_0^T \mathcal{H}_{n_k}(t,w_{n_k}(t))\dd
   t. 
  \end{align}
\end{thm}
\begin{proof}
This is shown in  \textsc{Stefanelli} \cite[Thm.\,4.3, pp.\,1626]{Stef08BEPD}.
\end{proof} 

As corollary of the previous theorem, we have the so-called
Fundamental Theorem for weak topologies which provides a
characterization of weak limits by \textsc{Young}-measures.

\begin{thm}\label{th:A2}
 (Fundamental Theorem for weak topologies)
  $\newline$ Let $1\leq p\leq \infty$ and let
  $(w_n)_{n\in\mathbb{N}}\subset \rmL^p(0,T;\mathcal{V})$ be a bounded
  sequence. If $p=1$, we suppose further that $(w_n)_{n\in\mathbb{N}}$
  is uniformly integrable in $\rmL^1(0,T;\mathcal{V})$. Then, there
  exists a subsequence $(w_{n_k})_{k\in\mathbb{N}}$ and a \textsc{Young}
  measure $\mathbold{\mu}=(\mu_t)_{t\in(0,T)}\in
  \mathscr{Y}(0,T;\mathcal{V})$ such that for almost every $t\in(0,T)$
  relation \eqref{eq:A3} holds and, setting
\begin{align}
\label{eq:A5}
\textsc{w}(t):=\int_\mathcal{V} w \, \dd \mu_t(w) \quad
\text{a.a. }t\in(0,T). 
\end{align} 
there holds
\begin{align}
\label{eq:A6}
w_{n_k}\rightharpoonup \textsc{w} \quad \text{in
}\rmL^p(0,T;\mathcal{V}) \quad \text{as }k\rightarrow \infty,
\end{align} with $\rightharpoonup$ replaced by $ \rightharpoonup^*$ if
$p=\infty$.
\end{thm}

\begin{lem}\label{le:A3}  
  Let the perturbed gradient system $(V,\calE,\Psi,B)$ satisfy the Assumptions
  \textnormal{(2.E), (2.$\Uppsi$)}, and 
  \textnormal{(2.B)} and let $u\in
  \AC(0,T;\mathcal{V})$ be an absolutely continuous curve such that
\begin{align}
\label{eq:A7}
  \partial \calE_t(u(t)) \neq \emptyset \quad \text{for a.a. }t\in(0,T), \quad
\text{and} \quad \sup_{t\in(0,T)}\calE_t(u(t))<+\infty. 
\end{align} 
Furthermore, let $\mathbold{\mu}=(\mu_t)_{t\in(0,T)}\in
\mathscr{Y}(0,T;\mathcal{V})$ be a \textsc{Young} measure such that
\begin{align}
\label{eq:A8}
\int_0^T \int_{V\times V^*\times \mathbb{R}} \left(
  \Psi_{u(t)}(v)+\Psi^*_{u(t)}(B(t,u(t))-\zeta)\right) \dd 
   \mu_t(v,\zeta,p)\dd t<+\infty, 
  \\
   u'(t)=\int_{V\times V^* \times \mathbb{R}} v \dd  
   \mu_t(v,\zeta,p) \quad \text{for a.a. }t\in (0,T),\notag
\end{align} 
and for almost all $t\in(0,T)$ for all $(v,\zeta,p)\in
\text{supp}(\mu_t)$ there holds $\zeta \in \partial \calE_t(u(t))$ and
$p\leq \partial_t \calE_t(u(t))$. Then,
\begin{align}
\label{eq:A9}
\begin{split}
  \text{the map $t\mapsto \calE_t(u(t)) $ is absolutely continuous on
    $(0,T)$, and}\\ 
  \frac{\rmd}{\rmd t}\calE_t(u(t))\geq \int_{V\times V^* \times
    \mathbb{R}} (\langle u'(t),\zeta\rangle+p)\dd \mu_t(v,\zeta,p)
  \quad \text{for a.a. }t\in(0,T).
\end{split}
\end{align}
\end{lem}
\begin{proof}
  This can be proven in exactly the same manner as in
  \cite[Prop.\,B.1, pp.\,305]{MiRoSa13NADN}.
\end{proof}

\begin{lem}\label{le:A4} (Measurable selection) Let the perturbed
  gradient system $(V,\calE,\Psi,B)$ satisfy the Assumptions
  \textnormal{(2.E), (2.$\Uppsi$)}, and 
  \textnormal{(2.B)}.  Furthermore,
  let $u\in \AC([0,T];V)$ be an absolutely continuous curve complying
  with \eqref{eq:A7}, and suppose that the set
\begin{align}
\label{eq:A10}
\mathcal{S}(t,u(t),u'(t)):=\lbrace &(\zeta,p)\in V^*\times \mathbb{R} \mid \zeta \in \partial \calE_t(u(t))\cap (B(t,u(t))-\partial \Psi_{u(t)}(u'(t)), \notag \\
& p\leq \partial_t \calE_t(u(t)) \rbrace \quad \text{ is non-empty for
  all }t\in(0,T). 
\end{align} Then, there exists measurable functions $\xi: (0,T)\rightarrow V^*, p:(0,T)\rightarrow \mathbb{R}$ such that 
\begin{align}
\label{eq:A11}
(\xi(t),p(t))\in \argmin\lbrace \Psi^*_{u(t)}(B(t,u(t))-\zeta)-p \mid (\zeta,p) \in \mathcal{S}(t,u(t),u'(t)) \rbrace 
\end{align} for a.a. $t\in(0,T)$.
\end{lem}
\begin{proof}
Like the previous result, this can also be proven in the same way as
\cite[Lem.\,B.2, pp.\,307]{MiRoSa13NADN}.
\end{proof}

\footnotesize

%

\newcommand{\etalchar}[1]{$^{#1}$}
\def\cprime{$'$}

\end{document}